\documentclass[english,10pt]{article}
\usepackage[T1]{fontenc}
\usepackage[latin9]{inputenc}
\usepackage{geometry}
\usepackage{bm}
\usepackage{amsmath}
\usepackage{amsthm}
\usepackage{amssymb}
\usepackage{setspace}
\usepackage[authoryear]{natbib}
\usepackage{hyperref}
\usepackage{graphicx,color}
\usepackage{comment}
\usepackage{stmaryrd}

\makeatletter
\theoremstyle{plain}
\newtheorem{lem}{\protect\lemmaname}
\theoremstyle{plain}
\newtheorem{prop}{\protect\propositionname}

\makeatother

\usepackage{babel}
\providecommand{\lemmaname}{Lemma}
\providecommand{\propositionname}{Proposition}

\begin{document}
\title{Finite sample inference for empirical Bayesian methods}
\author{Hien Duy Nguyen$^{1,2}$ \and Mayetri Gupta$^3$  }
\date{%
    $^1$ School of Mathematics and Physics, University of Queensland, St. Lucia 4067\\%
    $^2$ Department of Mathematics and Statistics, La Trobe University, Bundoora 3086\\%
    $^3$ School of Mathematics and Statistics, University of Glasgow, Glasgow G12 8QQ\\[2ex]%
}
\maketitle

\begin{abstract}
In recent years, empirical Bayesian (EB) inference has become an attractive approach for estimation in parametric models arising in a variety of real-life problems, especially in complex and high-dimensional scientific applications. However, compared to the relative abundance of available general methods for computing point estimators in the EB framework,  the construction of confidence sets and hypothesis tests with good theoretical properties remains difficult and problem specific. Motivated by the universal inference framework of \citet{Wasserman:2020aa},
we propose a general and universal method, based on holdout likelihood ratios, and utilizing the hierarchical structure of the specified Bayesian model for constructing confidence sets and hypothesis tests that are finite sample valid. We illustrate our method through a range of numerical studies and real data applications, which demonstrate that the approach is able to generate useful and meaningful inferential statements in the relevant contexts.
    
\end{abstract}

\section{\label{sec:Introduction}Introduction}

Let $\mathbf{D}_{n}=\left(\bm{X}_{i}\right)_{i\in\left[n\right]}$
be our data, presented as a sequence of $n\in\mathbb{N}=\left\{ 1,2,\dots\right\} $
random variables $\bm{X}_{i}\in \mathbb{X}$ ($i\in\left[n\right]=\left\{ 1,\dots,n\right\} $).
For each $i\in\left[n\right]$, let $\bm{\Theta}_{i}\in\mathbb{T}$
be a random variable with probability density function (PDF) $\pi\left(\bm{\theta}_{i};\bm{\psi}\right)$,
where $\bm{\psi}\in\mathbb{P}$ is a  hyperparameter. Furthermore,
suppose that $\left[ \bm{X}_{i}|\bm{\Theta}_{i}=\bm{\theta}_{i}\right] $
arises from a family of data generating processes (DGPs) with conditional
PDFs
\[
f\left(\bm{x}_{i}|\bm{\Theta}_{i}=\bm{\theta}_{i}\right)=f\left(\bm{x}_{i}|\bm{\theta}_{i}\right)\text{,}
\]
and that the sequence $\left(\left(\bm{X}_{i},\bm{\Theta}_{i}\right)\right)_{i\in\left[n\right]}$
is independent.

Suppose that $\left(\bm{\Theta}_{i}\right)_{i\in\left[n\right]}$
is realized at $\bm{\vartheta}_{n}^{*}=\left(\bm{\theta}_{i}^{*}\right)_{i\in\left[n\right]}$,
where each realization $\bm{\theta}_{i}^{*}$ ($i\in\left[n\right]$)
is unknown, and where $\bm{\psi}$ is also unknown. Let $\mathbb{I}\subset\left[n\right]$, and write $\bm{\vartheta}_{\mathbb{I}}^{*}=\left(\bm{\theta}_{i}^{*}\right)_{i\in\mathbb{I}}$.
When $\mathbb{I}=\left\{ i\right\} $, we shall use the shorthand
$\mathbb{I}=i$, where it causes no confusion. 

Under this setup, for significance level $\alpha\in\left(0,1\right)$,
we wish to draw inference regarding the realized sequence $\bm{\vartheta}_{n}^{*}$
by way of constructing $100\left(1-\alpha\right)\%$ confidence sets
$\mathcal{C}_{i}^{\alpha}\left(\mathbf{D}_{n}\right)$ that satisfy:
\begin{equation}
\text{Pr}_{\bm{\theta}_{i}^{*}}\left[\bm{\theta}_{i}^{*}\in\mathcal{C}_{i}^{\alpha}\left(\mathbf{D}_{n}\right)\right]\ge1-\alpha\text{,}\label{eq: good confidence}
\end{equation}
and $p$-values $P_{\mathbb{I}}\left(\mathbf{D}_{n}\right)$ for testing
null hypotheses $\text{H}_{0}:\bm{\vartheta}_{\mathbb{I}}^{*}\in\mathbb{T}_{\mathbb{I},0}\subset\mathbb{T}^{\left|\mathbb{I}\right|}$
that satisfy:
\begin{equation}
\sup_{\bm{\vartheta}_{\mathbb{I}}^{*}\in\mathbb{T}_{\mathbb{I},0}}\text{Pr}_{\bm{\vartheta}_{\mathbb{I}}^{*}}\left[P_{\mathbb{I}}\left(\mathbf{D}_{n}\right)\le\alpha\right]\le\alpha\text{,}\label{eq: good pvalue}
\end{equation}
where $\text{Pr}_{\bm{\theta}_{i}^{*}}$ and $\text{Pr}_{\bm{\vartheta}_{\mathbb{I}}^{*}}$
denote probability measures consistent with the PDF $f\left(\bm{x}_{i}|\bm{\theta}_{i}^{*}\right)$,
for each $i\in\left[n\right]$, and for all $i\in\mathbb{I}$, respectively. \textcolor{black}{That is, for a measurable set $\mathcal{A}\subset\mathbb{X}^{n}$, and
assuming absolute continuity of $\text{Pr}_{\bm{\vartheta}_{\mathbb{I}}^{*}}$
with respect to some measure $\mathfrak{m}$ (typically the Lebesgue
or counting measure), we can write
\begin{align}
\text{Pr}_{\bm{\vartheta}_{\mathbb{I}}^{*}}\left(\mathcal{A}\right) & =\int_{\mathcal{A}}\prod_{i\in\mathbb{I}}f\left(\bm{x}_{i}|\bm{\theta}_{i}^{*}\right)\prod_{j\notin\mathbb{I}}f\left(\bm{x}_{j}|\bm{\theta}_{j}\right)\text{d}\mathfrak{m}\left(\mathbf{d}_{n}\right)\text{,}\label{eq: Integral Definition}
\end{align}
where $\bm{\theta}_{j}$ is an arbitrary element of $\mathbb{T}$, for each $j\notin\mathbb{I}$.}

The setup above falls within the framework of empirical Bayesian (EB)
inference, as exposited in the volumes of \citet{Maritz:1989vp},
\citet{AhmedReid:2001ws}, \cite{Serdobolskii2008Multiparametric}, \citet{Efron2010}, and \citet{Bickel:2020wp}.
Over the years, there has been a sustained interest in the construction
and computation of EB point estimators for $\bm{\vartheta}_{n}^{*}$,
in various contexts, with many convenient and general computational tools now
made available, for instance, via the software of \citet{Johnstone:2005uq},
\citet{Leng:2013ug}, \citet{Koenker:2017tq}, and \citet{Narasimhan:2020ts}.
Unfortunately, the probabilistic properties of $\bm{\vartheta}_{n}^{*}$
tend to be difficult to characterize, making the construction of confidence
sets and hypothesis tests with good theoretical properties relatively less routine
than the construction of point estimators. When restricted to certain
classes of models, such constructions are nevertheless possible, as
exemplified by the works of \citet{Casella:1983ui}, \citet{Morris:1983uv},
\citet{Laird:1987wo}, \citet{Datta:2002wd}, \citet{Tai:2006wm},
\citet{Hwang:2009vr}, \citet{Hwang:2013vu}, and \citet{Yoshimori:2014vx},
among others.

In this work, we adapt the universal inference framework of \citet{Wasserman:2020aa}
to produce valid confidence sets and $p$-values with properties (\ref{eq: good confidence})
and (\ref{eq: good pvalue}), respectively, for arbitrary estimators
of $\bm{\vartheta}_{n}^{*}$. As with the constructions of \citet{Wasserman:2020aa},
the produced inferential methods  are all valid for finite sample size
$n$ and require no assumptions beyond correctness of model
specification. The confidence sets and $p$-values arise by construction
of holdout likelihood ratios that can be demonstrated to have the
$e$-value property, as described in \citet{Vovk:2021vi} (see also
the $s$-values of \citealp{Grunwald:2020vf} and the betting values
of \citealp{Shafer:2021vh}). Here, we are able to take into account
the hierarchical structure of the Bayesian specified model by using
the fact that parameterized $e$-values are closed when averaged with
respect to an appropriate probability measure (cf. \citealp{Vovk:2007aa}
and \citealp{Kaufmann2018Mixture-marting}). Due to the finite sample correctness of our constructions, we shall refer to our methods as finite sample EB (FSEB) techniques.

Along with our methodological developments, we also demonstrate the application of our FSEB techniques in numerical studies and real data applications. These applications include the use of FSEB methods for constructing confidence intervals (CIs) for the classic mean estimator of \cite{Stein1956Inadmissibility}, and testing and CI construction in Poisson--gamma models and Beta--binomial models, as per \cite{Koenker:2017tq} and \cite{Hardcastle13}, respectively.
Real data applications are demonstrated via the analysis of insurance data from \cite{Haastrup2000Comparison-of-s} and differential methylation data from \cite{Cruickshanks13}. In these real and synthetic applications, we show that FSEB methods, satisfying conditions (\ref{eq: good confidence}) and (\ref{eq: good pvalue}), are able to generate useful and meaningful inferential statements.

We proceed as follows. In Section \ref{sec:Confidence-sets-and},
we introduce the confidence set and $p$-value constructions for drawing
inference regarding EB models. In Section \ref{sec:Numerical-results},
numerical studies of simulated data are used to demonstrate the applicability and effectiveness of FSEB constructions. In Section \ref{sec:Results-applications}, FSEB methods are applied to real data to further show the practicality of the techniques. Lastly, in Section \ref{sec:Conclusion}, we
provide discussions and conclusions regarding our results.

\section{\label{sec:Confidence-sets-and}Confidence sets and hypothesis tests}

We retain the notation and setup from Section \ref{sec:Introduction}.
For each subset $\mathbb{I}\subset\left[n\right]$, let us write $\mathbf{D}_{\mathbb{I}}=\left(\bm{X}_{i}\right)_{i\in\mathbb{I}}$
and $\overline{\mathbf{D}}_{\mathbb{I}}=\left(\bm{X}_{i}\right)_{i\in\left[n\right]\backslash\mathbb{I}}$. 

Suppose that we have available some estimator of $\bm{\psi}$ that
only depends on $\overline{\mathbf{D}}_{\mathbb{I}}$ (and not $\mathbf{D}_{\mathbb{I}}$),
which we shall denote by $\hat{\bm{\psi}}_{\mathbb{I},n}$. Furthermore, for fixed $\bm{\psi}$, write the
integrated and unintegrated likelihood of the data $\mathbf{D}_{\mathbb{I}}$,
as
\begin{align}
\textcolor{black}{L_{\mathbb{I}}\left(\bm{\psi}\right)=\prod_{i\in\mathbb{I}}\int_{\mathbb{T}}f\left(\bm{X}_{i}|\bm{\theta}_{i}\right)\pi\left(\bm{\theta}_i;\bm{\psi}\right)\text{d}\mathfrak{n}(\bm{\theta}_{i}) \label{eq: Big L}}
\end{align}
and
\begin{align}
l_{\mathbb{I}}\left(\bm{\vartheta}_{\mathbb{I}}\right)=\prod_{i\in\mathbb{I}}f\left(\bm{X}_{i}|\bm{\theta}_{i}\right)\text{,} \label{eq: Little l}
\end{align}
respectively, where $\bm{\vartheta}_{\mathbb{I}}=\left(\bm{\theta}_{i}\right)_{i\in\mathbb{I}}$
(here, $\bm{\vartheta}_{\left\{ i\right\} }=\bm{\theta}_{i}$). \textcolor{black}{We note that in \eqref{eq: Big L}, we have assumed that $\pi(\cdot;\bm{\psi})$ is a density function with respect to some measure on $\mathbb{T}$, $\mathfrak{n}$.}

Define the ratio statistic:
\begin{align}
R_{\mathbb{I},n}\left(\bm{\vartheta}_{\mathbb{I}}\right)=L_{\mathbb{I}}\left(\hat{\bm{\psi}}_{\mathbb{I},n}\right)/l_{\mathbb{I}}\left(\bm{\vartheta}_{\mathbb{I}}\right)\text{,} \label{eq: Ratio}
\end{align}
and consider sets of the form 
\[
\mathcal{C}_{i}^{\alpha}\left(\mathbf{D}_{n}\right)=\left\{ \bm{\theta}\in\mathbb{T}:R_{i,n}\left(\bm{\theta}\right)\le1/\alpha\right\} \text{.}
\]
The following Lemma \textcolor{black}{is an adaptation of the main idea of \cite{Wasserman:2020aa} for the context of empierical Bayes estimators, and} allows us to show that $\mathcal{C}_{i}^{\alpha}\left(\mathbf{D}_{n}\right)$
satisfies property (\ref{eq: good confidence}).
\begin{lem}
\label{Lem: main}For each $\mathbb{I}\subset\left[n\right]$ and
fixed sequence $\bm{\vartheta}_{n}^{*}\in\mathbb{T}^{n}$, $\mathrm{E}_{\bm{\vartheta}_{\mathbb{I}}^{*}}\left[R_{\mathbb{I},n}\left(\bm{\vartheta}_{\mathbb{I}}^{*}\right)\right]=1$.
\end{lem}
\begin{proof}
\textcolor{black}{
Let $\mathbf{d}_{\mathbb{I}}$ and $\bar{\mathbf{d}}_{\mathbb{I}}$
be realizations of $\mathbf{D}_{\mathbb{I}}$ and $\overline{\mathbf{D}}_{\mathbb{I}}$,
respectively. Then, using (\ref{eq: Integral Definition}), write
\begin{align*}
\text{E}_{\bm{\theta}_{\mathbb{I}}^{*}}\left[R_{\mathbb{I},n}\left(\bm{\vartheta}_{\mathbb{I}}^{*}\right)\right] & =\int_{\mathbb{X}^{n}}R_{\mathbb{I},n}\left(\bm{\vartheta}_{\mathbb{I}}^{*}\right)\prod_{i\in\mathbb{I}}f\left(\bm{x}_{i}|\bm{\theta}_{i}^{*}\right)\prod_{j\notin\mathbb{I}}f\left(\bm{x}_{j}|\bm{\theta}_{j}\right)\text{d}\mathfrak{m}\left(\mathbf{d}_{n}\right)\\
 & \underset{\text{(i)}}{=}\int_{\mathbb{X}^{n-\left|\mathbb{I}\right|}}\int_{\mathbb{X}^{\left|\mathbb{I}\right|}}\frac{L_{\mathbb{I}}\left(\hat{\bm{\psi}}_{\mathbb{I},n}\right)}{l_{\mathbb{I}}\left(\bm{\vartheta}_{\mathbb{I}}^{*}\right)}\prod_{i\in\mathbb{I}}f\left(\bm{x}_{i}|\bm{\theta}_{i}^{*}\right)\text{d}\mathfrak{m}\left(\mathbf{d}_{\mathbb{I}}\right)\prod_{j\notin\mathbb{I}}f\left(\bm{x}_{j}|\bm{\theta}_{j}\right)\text{d}\mathfrak{m}\left(\bar{\mathbf{d}}_{\mathbb{I}}\right)\\
 & \underset{\text{(ii)}}{=}\int_{\mathbb{X}^{n-\left|\mathbb{I}\right|}}\int_{\mathbb{X}^{\left|\mathbb{I}\right|}}L_{\mathbb{I}}\left(\hat{\bm{\psi}}_{\mathbb{I},n}\right)\text{d}\mathfrak{m}\left(\mathbf{d}_{\mathbb{I}}\right)\prod_{j\notin\mathbb{I}}f\left(\bm{x}_{j}|\bm{\theta}_{j}\right)\text{d}\mathfrak{m}\left(\bar{\mathbf{d}}_{\mathbb{I}}\right)\\
 & \underset{\text{(iii)}}{=}\int_{\mathbb{X}^{n-\left|\mathbb{I}\right|}}\prod_{j\notin\mathbb{I}}f\left(\bm{x}_{j}|\bm{\theta}_{j}\right)\text{d}\mathfrak{m}\left(\bar{\mathbf{d}}_{\mathbb{I}}\right)\\
 & \underset{\text{(iv)}}{=}1\text{.}
\end{align*}
Here, (i) is true by definition of \eqref{eq: Ratio}, (ii) is true by definition of \eqref{eq: Little l}, (iii) is true by the fact that \eqref{eq: Big L} is a probability density function
on $\mathbb{X}^{\left|\mathbb{I}\right|}$, with respect to $\mathfrak{m}$,
and (iv) is true by the fact that $\prod_{j\notin\mathbb{I}}f\left(\bm{x}_{j}|\bm{\theta}_{j}\right)$
is a probability density function on $\mathbb{X}^{n-\left|\mathbb{I}\right|}$,
with respect to $\mathfrak{m}$.
}
\end{proof}
\begin{prop}
\label{Prop: Confidence}For each $i\in\left[n\right]$, $\mathcal{C}_{i}^{\alpha}\left(\mathbf{D}_{n}\right)$
is a $100\left(1-\alpha\right)\%$ confidence set, in the sense that
\[
\mathrm{Pr}_{\bm{\theta}_{i}^{*}}\left[\bm{\theta}_{i}^{*}\in\mathcal{C}_{i}^{\alpha}\left(\mathbf{D}_{n}\right)\right]\ge1-\alpha\text{.}
\]
\end{prop}
\begin{proof}
For each $i$, Markov's inequality states that
\[
\mathrm{Pr}_{\bm{\theta}_{i}^{*}}\left[R_{i,n}\left(\bm{\theta}_{i}^{*}\right)\ge1/\alpha\right]\le\alpha\text{E}_{\bm{\theta}_{i}^{*}}\left[R_{i,n}\left(\bm{\theta}_{i}^{*}\right)\right]=\alpha\text{,}
\]
which implies that
\[
\mathrm{Pr}_{\bm{\theta}_{i}^{*}}\left[\bm{\theta}_{i}^{*}\in\mathcal{C}_{i}^{\alpha}\left(\mathbf{D}_{n}\right)\right]=\mathrm{Pr}_{\bm{\theta}_{i}^{*}}\left[R_{i,n}\left(\bm{\theta}_{i}^{*}\right)\le1/\alpha\right]\ge1-\alpha
\]
by Lemma \ref{Lem: main}.
\end{proof}
Next, we consider the testing of null hypotheses $\text{H}_{0}\text{: }\bm{\vartheta}_{\mathbb{I}}^{*}\in\mathbb{T}_{\mathbb{I},0}$
against an arbitrary alternative $\text{H}_{1}\text{: }\bm{\vartheta}_{\mathbb{I}}^{*}\in\mathbb{T}_{\mathbb{I},1}\subseteq\mathbb{T}^{\left|\mathbb{I}\right|}$.
To this end, we define the maximum unintegrated likelihood estimator
of $\bm{\vartheta}_{\mathbb{I}}^{*}$, under $\text{H}_{0}$ as
\begin{equation}
\tilde{\bm{\vartheta}}_{\mathbb{I}}\in\left\{ \tilde{\bm{\vartheta}}_{\mathbb{I}}\in\mathbb{T}_{\mathbb{I},0}:l_{\mathbb{I}}\left(\tilde{\bm{\vartheta}}_{\mathbb{I}}\right)=\sup_{\bm{\vartheta}_{\mathbb{I}}\in\mathbb{T}_{\mathbb{I},0}}l_{\mathbb{I}}\left(\bm{\vartheta}_{\mathbb{I}}\right)\right\} \text{.}\label{eq: MLE}
\end{equation}

Using (\ref{eq: MLE}), and again letting $\hat{\bm{\psi}}_{\mathbb{I},n}$
be an arbitrary estimator of $\bm{\psi}$, depending only on $\overline{\mathbf{D}}_{\mathbb{I}}$,
we define the ratio test statistic
\[
T_{\mathbb{I}}\left(\mathbf{D}_{n}\right)=L_{\mathbb{I}}\left(\hat{\bm{\psi}}_{\mathbb{I},n}\right)/l_{\mathbb{I}}\left(\tilde{\bm{\vartheta}}_{\mathbb{I}}\right)\text{.}
\]
The following result establishes the fact that the $p$-value $P_{\mathbb{I}}\left(\mathbf{D}_{n}\right)=1/T_{\mathbb{I}}\left(\mathbf{D}_{n}\right)$
has the correct size, under $\text{H}_{0}$.
\begin{prop} \label{prop: test}
For any $\alpha\in\left(0,1\right)$ and $\bm{\vartheta}_{\mathbb{I}}^{*}\in\mathbb{T}_{\mathbb{I},0}$,
$\Pr_{\bm{\vartheta}_{\mathbb{I}}^{*}}\left[P_{\mathbb{I}}\left(\mathbf{D}_{n}\right)\le\alpha\right]\le\alpha$.
\end{prop}
\begin{proof}
Assume that $\bm{\vartheta}_{\mathbb{I}}^{*}\in\mathbb{T}_{\mathbb{I},0}$.
By Markov's inequality, we have
\begin{align*}
\mathrm{Pr}_{\bm{\vartheta}_{\mathbb{I}}^{*}}\left[T_{\mathbb{I}}\left(\mathbf{D}_{n}\right)\ge1/\alpha\right] & \le\alpha\text{E}_{\bm{\vartheta}_{\mathbb{I}}^{*}}\left[T_{\mathbb{I}}\left(\mathbf{D}_{n}\right)\right]\\
 & =\alpha\text{E}_{\bm{\vartheta}_{\mathbb{I}}^{*}}\left[\frac{L_{\mathbb{I}}\left(\hat{\bm{\psi}}_{\mathbb{I},n}\right)}{l_{\mathbb{I}}\left(\tilde{\bm{\vartheta}}_{\mathbb{I}}\right)}\right]\underset{\text{(i)}}{\le}\alpha\text{E}_{\bm{\vartheta}_{\mathbb{I}}^{*}}\left[\frac{L_{\mathbb{I}}\left(\hat{\bm{\psi}}_{\mathbb{I},n}\right)}{l_{\mathbb{I}}\left(\bm{\vartheta}_{\mathbb{I}}^{*}\right)}\right]\underset{\text{(ii)}}{=}\alpha\text{,}
\end{align*}
where the (i) is true due to the fact that $l_{\mathbb{I}}\left(\tilde{\bm{\vartheta}}_{\mathbb{I}}\right)\ge l_{\mathbb{I}}\left(\bm{\vartheta}_{\mathbb{I}}^{*}\right)$,
by the definition of (\ref{eq: MLE}), and the (ii) is true
due to Lemma \ref{Lem: main}.
\end{proof}

\textcolor{black}{We note that Propositions \ref{Prop: Confidence} and \ref{prop: test} are empirical Bayes analogues of Theorems 1 and 2 from \cite{Wasserman:2020aa}, which provide guarantees for universal inference confidence set and hypothesis test constructions, respectively. Furthermore, the use of Lemma \ref{Lem: main} in the proofs also imply that the CIs constructed via Proposition \ref{Prop: Confidence} are $e$-CIs, as defined by \cite{xu2022post}, and the $p$-values obtained via Proposition \ref{prop: test} can be said to be $e$-value calibrated, as per the definitions of \cite{wang2020false}.}

\section{\label{sec:Numerical-results} FSEB examples and some numerical results}

To demonstrate the usefulness of the FSEB results from Section \ref{sec:Confidence-sets-and}, we shall present a number of synthetic and real world applications of the confidence and testing constructions. All of the computation is conducted in the \textsf{R} programming environment (R Core Team, 2020) and replicable scripts are made available at \url{https://github.com/hiendn/Universal_EB}. Where unspecified, numerical optimization is conducted using the \texttt{optim()} or \texttt{optimize()} functions in the case of multivariate and univariate optimization, respectively.

\subsection{Stein's problem\label{subsec:Stein's-problem}}

We begin by studying the estimation of normal means, as originally
considered in \citet{Stein1956Inadmissibility}. Here, we largely
follow the exposition of \citet[Ch. 1]{Efron2010} and note that the estimator falls within the shrinkage paradigm exposited in \cite{Serdobolskii2008Multiparametric}. We consider this
setting due to its simplicity and the availability of a simple EB-based
method to compare our methodology against.

Let $\left(\left(X_{i},\Theta_{i}\right)\right)_{i\in\left[n\right]}$
be IID and for each $i\in\left[n\right]$, $\Theta_{i}\sim\text{N}\left(0,\psi^{2}\right)$
($\psi^{2}>0$) and $\left[X_{i}|\Theta_{i}=\theta_{i}\right]\sim\text{N}\left(\theta_{i},1\right)$,
where $\text{N}\left(\mu,\sigma^{2}\right)$ is the normal law with
mean $\mu\in\mathbb{R}$ and variance $\sigma^{2}>0$. We assume that
$\psi^{2}$ is unknown and that we observe data $\mathbf{D}_{n}$
and wish to construct CIs for the realizations $\theta_{n}^{*}$,
which characterize the DGP of the observations $X_{n}$.

Following \citet[Sec. 1.5]{Efron2010}, when $\psi^{2}$ is known,
the posterior distribution of $\left[\Theta_{n}|X_{n}=x_{n}\right]$ is $\text{N}\left(g\left(\psi^{2}\right)x_{n},g\left(\psi^{2}\right)\right)$,
where $g\left(\psi^{2}\right)=\psi^{2}/\left(1+\psi^{2}\right)$.
Using the data $\mathbf{D}_{n}$, we have the fact that $\sum_{i=1}^{n-1}X_{i}^{2}\sim\left(\psi^{2}+1\right)\chi_{n-1}^{2}$,
where $\chi_{\nu}^{2}$ is the chi-squared distribution with $\nu$
degrees of freedom. This implies a method-of-moment estimator for
$g$ of the form: $\bar{g}_{n}=1-\left(n-2\right)/\sum_{i=1}^{n}X_{i}^{2}$,
in the case of unknown $\psi^{2}$.

We can simply approximate the distribution of $\left[\Theta_{n}|\mathbf{D}_{n}\right]$
as $\text{N}\left(\bar{g}_{n} X_{n},\bar{g}_{n}\right)$, although this approximation
ignores the variability of $\bar{g}_{n}$. As noted by \citet[Sec. 1.5]{Efron2010},
via a hierarchical Bayesian interpretation using an objective Bayesian prior, we may instead
deduce the more accurate approximate distribution:
\begin{align}
\text{N}\left(\bar{g}_{n} X_{n},\bar{g}_{n}+2\left[X_{n}\left(1-\bar{g}_{n}\right)^{2}\right]/\left[n-2\right]\right). \label{eq: Morris approx}
\end{align}
\textcolor{black}{Specifically, \cite{Efron2010} considers the hyperparameter $\psi^{2}$
as being a random variable, say $\Psi^{2}$, and places a so-called
objective (or non-informative) prior on $\Psi^{2}$. In particular,
the improper prior assumption that $\Psi^{2}+1\sim\text{Uniform}\left(0,\infty\right)$
is made. Then, it follows from careful derivation that
\[
\text{E}\left[\Theta_{n}|\mathbf{D}_{n}\right]=\bar{g}_{n}X_{n}\text{ and }\text{var}\left[\Theta_{n}|\mathbf{D}_{n}\right]=\bar{g}_{n}+\frac{2X_{n}\left(1-\bar{g}_{n}\right)^{2}}{n-2}\text{,}
\]
and thus we obtain \eqref{eq: Morris approx} via a normal approximation for the distribution
of $\left[\Theta_{n}|\mathbf{D}_{n}\right]$ (cf. \citealt[Sec. 4]{morris1983parametric}).}

The approximation then provides $100\left(1-\alpha\right)\%$ posterior credible
intervals for $\Theta_{n}$ of the form
\begin{equation}
\bar{g}_n X_{n}\pm\zeta_{1-\alpha/2}\sqrt{\bar{g}_{n}+\frac{2\left[X_{n}\left(1-\bar{g}_{n}\right)^{2}\right]}{n-2}}\text{,}\label{eq: Morris Efron CI}
\end{equation}
where $\zeta_{1-\alpha/2}$ is the $\left(1-\alpha/2\right)$ quantile
of the standard normal distribution. This posterior result can then
be taken as an approximate $100\left(1-\alpha\right)\%$ confidence
interval for $\theta_{n}^{*}$.

Now, we wish to apply the FSEB results from Section \ref{sec:Confidence-sets-and}.
Here, $\mathbb{I}=\left\{ n\right\} $, and from the setup of the
problem, we have
\[
f\left(x_{n}|\theta_{n}\right)=\phi\left(x_{n};\theta_{n},1\right)\text{ and }\pi\left(\theta_{n};\psi\right)=\phi\left(\theta_{n};0,\psi^{2}\right)\text{,}
\]
where $\phi\left(x;\mu,\sigma^{2}\right)$ is the normal PDF with
mean $\mu$ and variance $\sigma^{2}$. Thus, 
\[
L_{\mathbb{I}}\left(\psi\right)=\int_{\mathbb{R}}\phi\left(X_{n};\theta,1\right)\phi\left(\theta;0,\psi^{2}\right)\text{d}\theta=\phi\left(X_{n};0,1+\psi^{2}\right)
\]
and $l_{\mathbb{I}}\left(\theta_{n}\right)=\phi\left(x_{n};\theta_{n},1\right)$,
which yields a ratio statistic of the form
\begin{align*}
R_{\mathbb{I},n}\left(\theta_{n}\right) & =L_{\mathbb{I}}\left(\psi_{-n}\right)/l_{\mathbb{I}}\left(\theta_{n}\right)\\
 & =\phi\left(X_{n};0,1+\hat{\psi}_{-n}^{2}\right)/\phi\left(X_{n};\theta_{n},1\right)\text{,}
\end{align*}
when combined with an appropriate estimator $\hat{\psi}_{-n}^{2}$
for $\psi^{2}$, using only $\bar{\mathbf{D}}_{\mathbb{I},n}=\mathbf{D}_{n-1}$.
We can obtain the region $\mathcal{C}_{\mathbb{I}}^{\alpha}\left(\mathbf{D}_{n}\right)$
by solving $R_{\mathbb{I},n}\left(\theta_{n}\right)\le1/\alpha$
to obtain:
\[
\left(X_{n}-\theta\right)^{2}\le2\log\left(1/\alpha\right)+2\log\left(1+\hat{\psi}_{-n}^{2}\right)+\frac{X_{n}^{2}}{\left(1+\hat{\psi}_{-n}^{2}\right)}\text{,}
\]
which, by Proposition \ref{Prop: Confidence}, yields the $100\left(1-\alpha\right)\%$
CI for $\theta_{n}^{*}$:
\begin{equation}
X_{n}\pm\sqrt{2\log\left(1/\alpha\right)+2\log\left(1+\hat{\psi}_{-n}^{2}\right)+\frac{X_{n}^{2}}{\left(1+\hat{\psi}_{-n}^{2}\right)}}\text{.}\label{eq: Wasserman Efron CI}
\end{equation}

We shall consider implementations of the CI of form (\ref{eq: Wasserman Efron CI})
using the estimator 
\[
\hat{\psi}_{-n}^{2}=\max\left\{ 0,s_{-n}^{2}-1\right\} \text{,}
\]
where $s_{-n}^{2}$ is the sample variance of the $\bar{\mathbf{D}}_{\mathbb{I},n}$,
and $s_{-n}^{2}-1$ is the method of moment estimator of $\psi^{2}$.
The maximum operator stops the estimator from becoming negative and
causes no problems in the computation of (\ref{eq: Wasserman Efron CI}).

We now compare the performances of the CIs of forms (\ref{eq: Morris Efron CI})
and (\ref{eq: Wasserman Efron CI}). To do so, we shall consider data
sets of sizes $n\in\left\{ 10,100,1000\right\} $, $\psi^{2}\in\left\{ 1^{2},5^{2},10^{2}\right\} $,
and $\alpha\in\left\{ 0.05,0.005,0.0005\right\} $. For each triplet
$\left(n,\psi^{2},\alpha\right)$, we repeat the computation of (\ref{eq: Morris Efron CI})
and (\ref{eq: Wasserman Efron CI}) $1000$ times and record the coverage
probability and average relative widths of the intervals (computed
as the width of (\ref{eq: Wasserman Efron CI}) divided by that of (\ref{eq: Morris Efron CI})).
The results of our experiment are presented in Table \ref{tab:Stein's-problem-simulation}.

\begin{table}
\caption{\label{tab:Stein's-problem-simulation}Stein's problem simulation
results reported as average performances over $1000$ replications. }

\begin{centering}
\begin{tabular}{llllll}
\hline 
$n$ & $\psi^{2}$ & $\alpha$ & Coverage of (\ref{eq: Morris Efron CI}) & Coverage of (\ref{eq: Wasserman Efron CI}) & Relative Width\tabularnewline
\hline 
\hline 
10 & $1^{2}$ & 0.05 & 0.948$^{*}$ & 1.000$^{*}$ & 1.979$^{*}$\tabularnewline
 &  & 0.005 & 0.988$^{*}$ & 1.000$^{*}$ & 1.738$^{*}$\tabularnewline
 &  & 0.0005 & 0.993$^{*}$ & 1.000$^{*}$ & 1.641$^{*}$\tabularnewline
 & $5^{2}$ & 0.05 & 0.943 & 1.000 & 1.902\tabularnewline
 &  & 0.005 & 0.994 & 1.000 & 1.543\tabularnewline
 &  & 0.0005 & 0.999 & 1.000 & 1.388\tabularnewline
 & $10^{2}$ & 0.05 & 0.947 & 1.000 & 2.058\tabularnewline
 &  & 0.005 & 0.994 & 1.000 & 1.633\tabularnewline
 &  & 0.0005 & 0.999 & 1.000 & 1.455\tabularnewline
\hline 
100 & $1^{2}$ & 0.05 & 0.937 & 0.999 & 2.068\tabularnewline
 &  & 0.005 & 0.997 & 1.000 & 1.806\tabularnewline
 &  & 0.0005 & 1.000 & 1.000 & 1.697\tabularnewline
 & $5^{2}$ & 0.05 & 0.949 & 1.000 & 1.912\tabularnewline
 &  & 0.005 & 0.995 & 1.000 & 1.540\tabularnewline
 &  & 0.0005 & 1.000 & 1.000 & 1.395\tabularnewline
 & $10^{2}$ & 0.05 & 0.947 & 1.000 & 2.068\tabularnewline
 &  & 0.005 & 0.995 & 1.000 & 1.635\tabularnewline
 &  & 0.0005 & 0.999 & 1.000 & 1.455\tabularnewline
\hline 
1000 & $1^{2}$ & 0.05 & 0.949 & 0.999 & 2.087\tabularnewline
 &  & 0.005 & 0.991 & 1.000 & 1.815\tabularnewline
 &  & 0.0005 & 1.000 & 1.000 & 1.705\tabularnewline
 & $5^{2}$ & 0.05 & 0.963 & 1.000 & 1.910\tabularnewline
 &  & 0.005 & 0.997 & 1.000 & 1.544\tabularnewline
 &  & 0.0005 & 1.000 & 1.000 & 1.399\tabularnewline
 & $10^{2}$ & 0.05 & 0.942 & 1.000 & 2.066\tabularnewline
 &  & 0.005 & 0.995 & 1.000 & 1.632\tabularnewline
 &  & 0.0005 & 0.999 & 1.000 & 1.455\tabularnewline
\hline 
\end{tabular}
\par\end{centering}
$^{*}$The results on these lines are computed from 968, 967, and
969 replicates, respectively, from top to bottom. This was due to the negative estimates
of the standard error in the computation of (\ref{eq: Morris Efron CI}). 
\end{table}

From Table \ref{tab:Stein's-problem-simulation}, we observe that
the CIs of form (\ref{eq: Morris Efron CI}) tended to produce intervals
with the desired levels of coverage, whereas the FSEB CIs of form (\ref{eq: Wasserman Efron CI})
tended to be conservative and contained the parameter of interest
in almost all replications. The price that is paid for this conservativeness
is obvious when viewing the relative widths, which implies that for
$95\%$ CIs, the EB CIs of form (\ref{eq: Wasserman Efron CI}) are twice as wide,
on average, when compared to the CIs of form (\ref{eq: Morris Efron CI}).
However, the relative widths decrease as $\alpha$ gets smaller, implying
that the intervals perform relatively similarly when a high level of
confidence is required. We further observe that $n$ and $\psi^{2}$
had little effect on the performances of the intervals except in the
case when $n=10$ and $\psi^{2}=1$, whereupon it was possible for the
intervals of form (\ref{eq: Morris Efron CI}) to not be computable
in some cases.

From these results we can make a number of conclusions. Firstly, if
one is willing to make the necessary hierarchical
and objective Bayesian assumptions, \textcolor{black}{as stated in \citet[Sec. 1.5]{Efron2010}}, then the intervals of form (\ref{eq: Morris Efron CI})
provide very good performance. However, without those assumptions,
we can still obtain reasonable CIs that have correct coverage via
the FSEB methods from Section \ref{sec:Confidence-sets-and}. Furthermore,
these intervals become more efficient compared to (\ref{eq: Morris Efron CI}) when
higher levels of confidence are desired. Lastly, when $n$ is small
and $\psi^{2}$ is also small, the intervals of form (\ref{eq: Morris Efron CI})
can become uncomputable and thus one may consider the use of (\ref{eq: Wasserman Efron CI})
as an alternative.

\subsection{Poisson--gamma count model}

The following example is taken from \citet{Koenker:2017tq} and was
originally studied in \citet{Norberg1989Experience-rati} and then
subsequently in \citet{Haastrup2000Comparison-of-s}. In this example,
we firstly consider IID parameters $\left(\Theta_{i}\right)_{i\in\left[n\right]}$
generated with gamma DGP: $\Theta_{i}\sim\text{Gamma}\left(a,b\right)$,
for each $i\in\left[n\right]$, where $a>0$ and $b>0$ are the shape
and rate hyperparameters, respectively, which we put into $\bm{\psi}$.
Then, for each $i$, we suppose that the data $\mathbf{D}_{n}=\left(X_{i}\right)_{i\in\left[n\right]}$,
depending on the covariate sequence $\mathbf{w}_{n}=\left(w_{i}\right)_{i\in\left[n\right]}$,
has the Poisson DGP: $\left[X_{i}|\Theta_{i}=\theta_{i}\right]\sim\text{Poisson}\left(\theta_{i}w_{i}\right)$,
where $w_{i}>0$. We again wish to use the data $\mathbf{D}_{n}$
to estimate the realization of $\Theta_{n}$: $\theta_{n}^{*}$, which
characterizes the DGP of $X_{n}$.

Under the specification above, for each $i$, we have the fact that
$\left(X_{i},\Theta_{i}\right)$ has the joint PDF:
\begin{equation}
f\left(x_{i},\theta_{i};\bm{\psi}\right)=\frac{b^{a}}{\Gamma\left(a\right)}\theta_{i}^{a-1}\exp\left(b\theta_{i}\right)\frac{\left(\theta_{i}w_{i}\right)^{x_{i}}\exp\left(-\theta_{i}w_{i}\right)}{x_{i}}\text{,}\label{eq: joint poisson-gamma}
\end{equation}
which we can marginalize to obtain
\begin{equation}
f\left(x_{i};\bm{\psi}\right)={x_{i}+a+1 \choose x_{i}}\left(\frac{b}{w_{i}+b}\right)^{a}\left(\frac{w_{i}}{w_{i}+b}\right)^{x_{i}}\text{,}\label{eq: marginal poisson-gamma}
\end{equation}
and which can be seen as a Poisson--gamma mixture model. We can then
construct the likelihood of $\mathbf{D}_{n}$ using expression (\ref{eq: marginal poisson-gamma}),
from which we may compute maximum likelihood estimates $\hat{\bm{\psi}}_{n}=\left(\hat{a}_{n},\hat{b}_{n}\right)$
of $\bm{\psi}$. Upon noting that (\ref{eq: joint poisson-gamma})
implies the conditional expectation $\text{E}\left[\Theta_{i}|X_{i}=x_{i}\right]=\left(x_{i}+a\right)/\left(w_{i}+b\right)$,
we obtain the estimator for $\theta_{n}^{*}$:
\begin{equation}
\hat{\theta}_{n}=\frac{X_{i}+\hat{a}_{n}}{w_{i}+\hat{b}_{n}}\text{.}\label{eq: poisson--gamma estimator}
\end{equation}

\subsubsection{Confidence intervals\label{subsec:Confidence-intervals-Poisson}}

We again wish to apply the general result from Section \ref{sec:Confidence-sets-and}
to construct CIs. Firstly, we have $\mathbb{I}=\left\{ n\right\} $
and 
\[
f\left(x_{n}|\theta_{n}\right)=\frac{\left(\theta_{n}w_{n}\right)^{x_{n}}\exp\left(-\theta_{n}w_{n}\right)}{x_{n}}\text{ and }\pi\left(\theta_{n};\bm{\psi}\right)=\frac{b^{a}}{\Gamma\left(a\right)}\theta_{n}^{a-1}\exp\left(b\theta_{n}\right)\text{.}
\]
As per (\ref{eq: marginal poisson-gamma}), we can write
\[
L_{\mathbb{I}}\left(\bm{\psi}\right)={X_{n}+a+1 \choose X_{n}}\left(\frac{b}{w_{n}+b}\right)^{a}\left(\frac{w_{n}}{w_{n}+b}\right)^{X_{n}}\text{.}
\]
Then, since $l_{\mathbb{I}}\left(\theta_{n}\right)=f\left(X_{n}|\theta_{n}\right)$,
we have
\begin{align*}
R_{\mathbb{I},n}\left(\theta_{n}\right) & =L_{\mathbb{I}}\left(\bm{\psi}\right)/l_{\mathbb{I}}\left(\theta_{n}\right)\\
 & ={X_{n}+\hat{a}_{-n}+1 \choose X_{n}}\left(\frac{\hat{b}_{-n}}{w_{n}+\hat{b}_{-n}}\right)^{\hat{a}_{-n}}\left(\frac{w_{n}}{w_{n}+\hat{b}_{-n}}\right)^{X_{n}}\frac{X_{n}}{\left(\theta_{n}w_{n}\right)^{X_{n}}\exp\left(-\theta_{n}w_{n}\right)}\text{,}
\end{align*}
when combined with an estimator $\hat{\bm{\psi}}_{-n}=\left(\hat{a}_{-n},\hat{b}_{-n}\right)$
of $\bm{\psi}$, using only $\bar{\mathbf{D}}_{\mathbb{I},n}=\mathbf{D}_{n-1}$. 

For any $\alpha\in\left(0,1\right)$, we then obtain a $100\left(1-\alpha\right)\%$
CI for $\theta_{n}$ by solving $R_{\mathbb{I},n}\left(\theta_{n}\right)\le1/\alpha$,
which can be done numerically. We shall use the MLE of $\bm{\psi}$,
computed with the data $\bar{\mathbf{D}}_{\mathbb{I},n}$ and marginal PDF (\ref{eq: marginal poisson-gamma}),
as the estimator $\hat{\bm{\psi}}_{-n}$.

To demonstrate the performance of the CI construction, above, we conduct
the following numerical experiment. We generate data sets consisting
of $n\in\left\{ 10,100,1000\right\} $ observations characterized by
hyperparameters $\bm{\psi}=\left(a,b\right)=\left\{ \left(2,2\right),\left(2,5\right),\left(5,2\right)\right\} $,
and we compute intervals using significance levels $\alpha\in\left\{ 0.05,0.005,0.0005\right\} $.
Here, we shall generate $\mathbf{w}_{n}$ IID uniformly between 0
and 10. For each triplet $\left(n,\bm{\psi},\alpha\right)$, we repeat
the construction of our CIs 1000 times and record the coverage probability
and average width for each case. The results of the experiment are
reported in Table \ref{tab:Poisson-gamma mixture sim - conf}.

\begin{table}
\caption{\label{tab:Poisson-gamma mixture sim - conf}Experimental results
for CIs constructed for Poisson--gamma count models. The Coverage
and Length columns report the coverage proportion and average lengths
in each scenario, as computed from 1000 replications.}

\begin{centering}
\begin{tabular}{lllll}
\hline 
$n$ & $\bm{\psi}$ & $\alpha$ & Coverage & Length\tabularnewline
\hline 
\hline 
10 & $\left(2,2\right)$ & 0.05 & 0.998 & 3.632\tabularnewline
 &  & 0.005 & 1.000 & 5.484\tabularnewline
 &  & 0.0005 & 1.000 & 6.919\tabularnewline
 & $\left(2,5\right)$ & 0.05 & 0.999 & 2.976\tabularnewline
 &  & 0.005 & 0.999 & 3.910\tabularnewline
 &  & 0.0005 & 1.000 & 5.481\tabularnewline
 & $\left(5,2\right)$ & 0.05 & 0.997$^{*}$ & 5.468$^{*}$\tabularnewline
 &  & 0.005 & 0.999$^{*}$ & 7.118$^{*}$\tabularnewline
 &  & 0.0005 & 1.000$^{*}$ & 8.349$^{*}$\tabularnewline
\hline 
100 & $\left(2,2\right)$ & 0.05 & 0.998 & 3.898\tabularnewline
 &  & 0.005 & 0.999 & 5.277\tabularnewline
 &  & 0.0005 & 1.000 & 6.883\tabularnewline
 & $\left(2,5\right)$ & 0.05 & 0.999 & 2.958\tabularnewline
 &  & 0.005 & 1.000 & 3.914\tabularnewline
 &  & 0.0005 & 1.000 & 5.374\tabularnewline
 & $\left(5,2\right)$ & 0.05 & 1.000 & 5.628\tabularnewline
 &  & 0.005 & 1.000 & 7.124\tabularnewline
 &  & 0.0005 & 1.000 & 8.529\tabularnewline
\hline 
1000 & $\left(2,2\right)$ & 0.05 & 1.000 & 4.070\tabularnewline
 &  & 0.005 & 1.000 & 5.424\tabularnewline
 &  & 0.0005 & 1.000 & 6.344\tabularnewline
 & $\left(2,5\right)$ & 0.05 & 0.999 & 3.049\tabularnewline
 &  & 0.005 & 1.000 & 3.960\tabularnewline
 &  & 0.0005 & 1.000 & 5.479\tabularnewline
 & $\left(5,2\right)$ & 0.05 & 0.998 & 5.297\tabularnewline
 &  & 0.005 & 1.000 & 7.205\tabularnewline
 &  & 0.0005 & 1.000 & 8.714\tabularnewline
\hline 
\end{tabular}
\par\end{centering}
$^{*}$The results on these lines are computed from 999, 999, and
998 replicates, respectively. This was due to there being no solutions
to the inequality $R_{\mathbb{I},n}\left(\theta_{n}\right)\le1/\alpha$,
with respect to $\theta_{n}>0$ in some cases.
\end{table}

From Table \ref{tab:Poisson-gamma mixture sim - conf}, we observe
that the empirical coverage of the CIs are higher than the nominal
value and are thus behaving as per the conclusions of Proposition
\ref{Prop: Confidence}. As expected, we also find that increasing
the nominal confidence level also increases the coverage proportion,
but at a cost of increasing the lengths of the CIs. From the usual
asymptotic theory of maximum likelihood estimators, we anticipate
that increasing $n$ will decrease the variance of the estimator $\hat{\bm{\psi}}_{-n}$.
However, as in Section \ref{subsec:Stein's-problem}, this does not
appear to have any observable effect on either the coverage proportion
nor lengths of the CIs.

\subsubsection{Hypothesis tests}

Next, we consider testing the null hypothesis $\text{H}_{0}$: $\theta_{n-1}^{*}=\theta_{n}^{*}$.
To this end, we use the hypothesis testing framework from Section
\ref{sec:Confidence-sets-and}. That is, we let $\mathbb{I}=\left\{ n-1,n\right\} $
and estimate $\bm{\psi}$ via the maximum likelihood estimator $\hat{\bm{\psi}}_{\mathbb{I},n}=\left(a_{\mathbb{I},n},b_{\mathbb{I},n}\right)$,
computed from the data $\bar{\mathbf{D}}_{\mathbb{I},n}=\mathbf{D}_{n-2}$. 

We can write
\[
L_{\mathbb{I}}\left(\hat{\bm{\psi}}_{\mathbb{I},n}\right)=\prod_{i=n-1}^{n}{X_{i}+a_{\mathbb{I},n}+1 \choose X_{i}}\left(\frac{b_{\mathbb{I},n}}{w_{i}+b_{\mathbb{I},n}}\right)^{a_{\mathbb{I},n}}\left(\frac{w_{i}}{w_{i}+b_{\mathbb{I},n}}\right)^{X_{i}}\text{,}
\]
\[
l_{\mathbb{I}}\left(\bm{\vartheta}_{\mathbb{I}}^{*}\right)=\prod_{i=n-1}^{n}\frac{\left(\theta_{i}^{*}w_{i}\right)^{X_{n}}\exp\left(-\theta_{i}^{*}w_{i}\right)}{X_{i}}\text{,}
\]
and $\bm{\vartheta}_{\mathbb{I}}^{*}=\left(\theta_{n-1}^{*},\theta_{n}^{*}\right)$.
We are also required to compute the maximum likelihood estimator of
$\bm{\vartheta}_{\mathbb{I}}^{*}$, under $\text{H}_{0}$, as per
(\ref{eq: MLE}), which can be written as
\[
\tilde{\bm{\vartheta}}_{\mathbb{I}}\in\left\{ \tilde{\bm{\theta}}=\left(\theta,\theta\right):l_{\mathbb{I}}\left(\tilde{\bm{\theta}}\right)=\sup_{\theta>0}\text{ }\prod_{i=n-1}^{n}\frac{\left(\theta w_{i}\right)^{X_{n}}\exp\left(-\theta w_{i}\right)}{X_{i}}\right\} \text{.}
\]
Using the components above, we define the test statistic $T_{\mathbb{I}}\left(\mathbf{D}_{n}\right)=L_{\mathbb{I}}\left(\hat{\bm{\psi}}_{\mathbb{I},n}\right)/l_{\mathbb{I}}\left(\tilde{\bm{\vartheta}}_{\mathbb{I}}\right)$,
from which we can derive the $p$-value $P_{\mathbb{I}}\left(\mathbf{D}_{n}\right)=1/T_{\mathbb{I}}\left(\mathbf{D}_{n}\right)$
for testing $\text{H}_{0}$.

To demonstrate the application of this test, we conduct another numerical
experiment. As in Section \ref{subsec:Confidence-intervals-Poisson},
we generate data sets of sizes $n\in\left\{ 10,100,1000\right\} $,
where the data $\mathbf{D}_{n-1}$ are generated with parameters $\left(\Theta_{i}\right)_{i\in\left[n-1\right]}$
arising from gamma distributions with hyperparameters $\bm{\psi}=\left(a,b\right)=\left\{ \left(2,2\right),\left(2,5\right),\left(5,2\right)\right\} $.
The final observation $X_{n}$, making up $\mathbf{D}_{n}$, is then
generated with parameter $\Theta_{n}=\Theta_{n-1}+\Delta$, where $\Delta\in\left\{ 0,1,5,10\right\} $.
As before, we generate the covariate sequence $\mathbf{w}_{n}$ IID
uniformly between 0 and 10. For each triplet $\left(n,\bm{\psi},\Delta\right)$,
we test $\text{H}_{0}$: $\theta_{n-1}^{*}=\theta_{n}^{*}$ 1000 times
and record the average number of rejections under at the levels of
significance $\alpha\in\left\{ 0.05,0.005,0.0005\right\} $. The results
are then reported in Table \ref{tab:Poisson-gamma mixture sim - test}.

\begin{table}
\caption{\label{tab:Poisson-gamma mixture sim - test}Experimental results
for testing the hypothesis $\text{H}_{0}$: $\theta_{n-1}^{*}=\theta_{n}^{*}$
for Poisson--gamma count models. The Rejection Proportion columns
report the average number of rejections, from 1000 tests, at levels
of significance $\alpha\in\left\{ 0.05,0.005,0.0005\right\} $.}

\centering{}%
\begin{tabular}{llllll}
\hline 
 &  &  & \multicolumn{3}{l}{Rejection Proportion at level $\alpha$}\tabularnewline
$n$ & $\bm{\psi}$ & $\Delta$ & $0.05$ & $0.005$ & $0.0005$\tabularnewline
\hline 
\hline 
10 & $\left(2,2\right)$ & 0 & 0.000 & 0.000 & 0.000\tabularnewline
 &  & 1 & 0.004 & 0.000 & 0.000\tabularnewline
 &  & 5 & 0.280 & 0.193 & 0.128\tabularnewline
 &  & 10 & 0.413 & 0.363 & 0.317\tabularnewline
 & $\left(2,5\right)$ & 0 & 0.000 & 0.000 & 0.000\tabularnewline
 &  & 1 & 0.007 & 0.002 & 0.000\tabularnewline
 &  & 5 & 0.143 & 0.096 & 0.064\tabularnewline
 &  & 10 & 0.222 & 0.192 & 0.170\tabularnewline
 & $\left(5,2\right)$ & 0 & 0.001 & 0.000 & 0.000\tabularnewline
 &  & 1 & 0.001 & 0.000 & 0.000\tabularnewline
 &  & 5 & 0.177 & 0.107 & 0.052\tabularnewline
 &  & 10 & 0.389 & 0.320 & 0.254\tabularnewline
\hline 
100 & $\left(2,2\right)$ & 0 & 0.000 & 0.000 & 0.000\tabularnewline
 &  & 1 & 0.014 & 0.003 & 0.000\tabularnewline
 &  & 5 & 0.401 & 0.289 & 0.194\tabularnewline
 &  & 10 & 0.562 & 0.489 & 0.427\tabularnewline
 & $\left(2,5\right)$ & 0 & 0.000 & 0.000 & 0.000\tabularnewline
 &  & 1 & 0.015 & 0.000 & 0.000\tabularnewline
 &  & 5 & 0.208 & 0.127 & 0.074\tabularnewline
 &  & 10 & 0.296 & 0.235 & 0.179\tabularnewline
 & $\left(5,2\right)$ & 0 & 0.000 & 0.000 & 0.000\tabularnewline
 &  & 1 & 0.004 & 0.000 & 0.000\tabularnewline
 &  & 5 & 0.264 & 0.150 & 0.090\tabularnewline
 &  & 10 & 0.500 & 0.425 & 0.344\tabularnewline
\hline 
1000 & $\left(2,2\right)$ & 0 & 0.001 & 0.000 & 0.000\tabularnewline
 &  & 1 & 0.021 & 0.001 & 0.000\tabularnewline
 &  & 5 & 0.423 & 0.300 & 0.216\tabularnewline
 &  & 10 & 0.576 & 0.513 & 0.450\tabularnewline
 & $\left(2,5\right)$ & 0 & 0.000 & 0.000 & 0.000\tabularnewline
 &  & 1 & 0.012 & 0.000 & 0.000\tabularnewline
 &  & 5 & 0.185 & 0.108 & 0.061\tabularnewline
 &  & 10 & 0.321 & 0.254 & 0.197\tabularnewline
 & $\left(5,2\right)$ & 0 & 0.000 & 0.000 & 0.000\tabularnewline
 &  & 1 & 0.003 & 0.001 & 0.000\tabularnewline
 &  & 5 & 0.276 & 0.168 & 0.088\tabularnewline
 &  & 10 & 0.507 & 0.428 & 0.354\tabularnewline
\hline 
\end{tabular}
\end{table}

The results for the $\Delta=0$ cases in Table \ref{tab:Poisson-gamma mixture sim - test}
show that the tests reject true null hypotheses at below the nominal
sizes $\alpha$, in accordance with Proposition \ref{prop: test}.
For each combination of $n$ and $\bm{\psi}$, as $\Delta$ increases,
the proportion of rejections increase, demonstrating that the tests
become more powerful when detecting larger differences between $\theta_{n-1}^{*}$
and $\theta_{n}^{*}$, as expected. There also appears to be an increase
in power due to larger sample sizes. This is an interesting outcome,
since we can only be sure that sample size affects the variability
of the estimator $\bm{\psi}_{\mathbb{I},n}$. Overall, we can be
confident that the tests are behaving as required, albeit they may
be somewhat underpowered  as they are not achieving the nominal
sizes.

\subsection{\label{sec:sim-betabinom} Beta--binomial data series }

Data from genome-level biological studies,  using modern high-throughput sequencing technologies \citep{Krueger12},  often take the form of a series of counts, which may be modelled through sets of non-identical (possibly correlated) binomial distributions, with beta priors, in a Bayesian framework. The question of interest may vary, for example, from assessing the range of likely values for the binomial parameter in a particular region of the data, to comparing whether two sections of one or more data series are generated from identical distributions. For purposes of demonstrating the performance of the FSEB method in these scenario,  we will make the simplifying assumption that all data points are independently distributed, within, as well as across, any of $G$ data series that may be observed.

\subsubsection{\label{sec:betabinomCI}Confidence Sets}

First, let us assume that we only have a single series, i.e. $G=1$. Then, we can assume $X_{i} \sim \mbox{Bin}(m_{i}, \theta_{i})$, and propose a common prior distribution for
$\Theta_{i}$ $(i=1,\ldots, n)$: Beta($\gamma, \beta$). Using the techniques described in Section
\ref{sec:Confidence-sets-and}, we can find confidence sets for
$\theta^{*}_{i}$, $(i=1,
\ldots, n)$. For each  $i$, we define, as
previously, a subset $\mathbb{I} = \{i\}$, so that
$\mathbf{D}_{\mathbb{I}}= X_{i}$
and
$\overline{\mathbf{D}}_{\mathbb{I}}=\left(X_{i}\right)_{i\in\left[n\right]\backslash\{i\}}$. We
then have, 
\[
  R_{\mathbb{I},n}\left(\bm{\vartheta}_{\mathbb{I}}\right)
  = \frac{L_{\mathbb{I}}\left(\hat{\bm{\psi}}_{\mathbb{I},n}\right)}{l_{\mathbb{I}}\left(\bm{\vartheta}_{\mathbb{I}}\right)},
  \]
  where 
  \begin{align*}
    l_{\mathbb{I}}\left(\bm{\vartheta}_{\mathbb{I}}\right) &=  \binom{m_{i}}{ x_{i}}
\theta_i^{x_{i}} (1 - \theta_i)^{m_{i}-x_{i}}
\end{align*}
and
\begin{align*}
    L_{\mathbb{I}}\left(\hat{\bm{\psi}}_{\mathbb{I},n}\right)   &=
\int_{\theta_{i}} f(x_{i} |
\theta_{i}) \pi(\theta_{i};\; \hat{\gamma}_{-n}, \hat{\beta}_{-n})
\text{d} \theta_{i},
\end{align*}
    which gives the ratio
    \begin{align}
      R_{\mathbb{I},n}\left(\bm{\vartheta}_{\mathbb{I}}\right) &=
      \frac{B(x_{i} +
  \hat{\gamma}_{-n}, m_{i}-x_{i} +\hat{\beta}_{-n}
  )}{B(\hat{\gamma}, \hat{\beta}_{-n}) \theta_{i}^{x_{i}} (1 - \theta_{i})^{m_{i}-x_{i}} }.
\label{eq:betabinR}
    \end{align}
Here, $\hat{\gamma}_{-n}$ and $\hat{\beta}_{-n}$ are the empirical Bayes estimates of
 $\gamma$ and $\beta$, given by
 \begin{align*}
   \hat{\gamma}_{-n} &= (\hat{\phi}_\text{EB}^{-1} - 1) \hat{\mu}_\text{EB}
\end{align*}
and
\begin{align*}
    \hat{\beta}_{-n} &=    (\hat{\phi}_\text{EB}^{-1} - 1) (1 - \hat{\mu}_\text{EB}) ,          
\end{align*}
  where  
\begin{align*}
  \hat{\mu}_\text{EB} &= \frac{1}{n-1} \sum_{j\in \left[n\right]\backslash i} \frac{x_{j}}{m_{j}},\\
  \hat{\phi}_\text{EB} &= \left[\frac{\bar{m} \hat{V}_x}{\mu(1-\mu)} - 1 \right]\bigg{/} (\bar{m} -1),
\end{align*}
 $\bar{m} = \frac{1}{n-1} \sum_{j\in \left[n\right]\backslash i} m_{j}$, and $\hat{V}_x =
\frac{1}{n-1} \sum_{j\in \left[n\right]\backslash i} (\frac{x_{j}}{m_{j}} - \hat{\mu}_\text{EB})^2$. Further, $B\left(a,b\right)=\int_{0}^{1}t^{a-1}\left(1-t\right)^{b-1}\text{d}t$ is the Beta function, taking inputs $a>0$ and $b>0$.

We simulated data from the binomial model under two cases: (a) setting
beta hyperparameters $(\alpha, \beta) = (10,10)$, and hierarchically simulating
$\theta_{i}^{*}$, $i\in[n]$, and then $x_{i}$ from a binomial distribution;
and (b) setting  a range of $\theta_i^{*}$ ($i\in[n]$) values equidistantly spanning the interval $(0.1,
0.9)$ for $n=10, 100$. Here, $m_i$
($i\in[n]$) were given integer values uniformly generated in
the range $[15, 40]$. In
all cases, it was seen that the CIs had
perfect coverage, always containing the true value of $\theta_{i}^{*}$. An example of the $n=10$ case is shown in   Figure \ref{fig-betabinomCI-a}.

\begin{figure}
  \begin{center}
\includegraphics[width=0.7\textwidth]{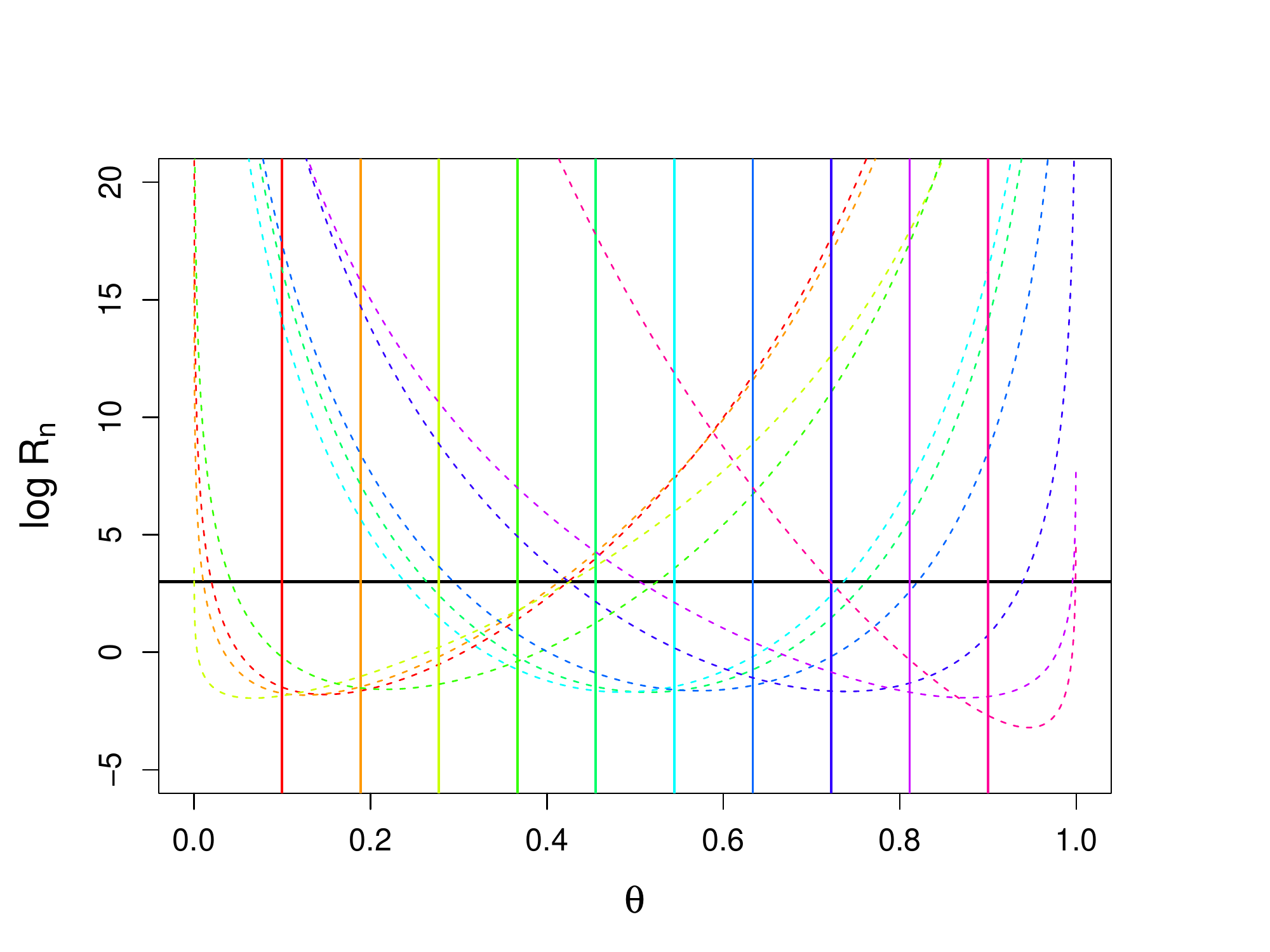}

\caption{\label{fig-betabinomCI-a}
      Plots of 95\% confidence regions for
      $\theta_i^{*}$ when true
      values of $\theta_i^{*}$  span the interval 0.1 to 0.9 ($n=10$). Here, the $95\%$
      CIs are given by the points where the curves
      for $\log R_{\mathbb{I},n}\left(\bm{\vartheta}_{\mathbb{I}}\right)$
      intersect with the horizontal line (black), representing a
      confidence level of $1-\alpha=0.95$. Each CI can
      be seen to contain the corresponding true value of $\theta_i^{*}$,
      represented by a vertical line of the same colour as the interval.
  } 
  \end{center}
  \end{figure}

\subsubsection{Hypothesis testing}
  
Aiming to detect genomic regions that may have differing characteristics between two series, a  pertinent
question of interest may be considered by testing the hypotheses: $H_0$:
$\theta_{i1}^{*} = \theta_{i2}^{*}$ vs. $H_1$: $\theta_{i1}^{*} \neq \theta_{i2}^{*}$,
for every $i \in [n]$ (with $G=2$ series). Then, $\mathbf{D}_{n}
=\left(\bm{X}_{i}\right) _{i\in\left[n\right]}$, where $\bm{X}_{i} =
(X_{i1}, X_{i2})$.
From Section \ref{sec:Confidence-sets-and}, the ratio test
statistic takes the form
\[
T_{\mathbb{I}}\left(\mathbf{D}_{n}\right) =
L_{\mathbb{I}}\left(\hat{\gamma}_{\mathbb{I},n}, \hat{\beta}_{\mathbb{I},n}\right)/l_{\mathbb{I}}\left(\tilde{\bm{\vartheta}}_{\mathbb{I}}\right)\text{,}
  \]
where $\hat{\gamma}_{\mathbb{I},n}$ and
$\hat{\beta}_{\mathbb{I},n}$ are EB estimators of $\gamma$ and
$\beta$, depending only on  $\bar{\mathbf{D}}_{\mathbb{I},n} =
\mathbf{D}_{n}\backslash \{X_{i1}, X_{i2}\}$. With
$\tilde{{\vartheta}}_{\mathbb{I}} = \frac{x_{i1} + x_{i2}}{m_{i1}+m_{i2}}  = \tilde{\theta}_i$, write
$
l_{\mathbb{I}}\left(\tilde{{\vartheta}}_{\mathbb{I}}\right)
= f( x_{i1}, x_{i2}| \tilde{\theta}_i)$, and
\begin{align*}
L_{\mathbb{I}}\left(\hat{\gamma}_{\mathbb{I},n}, \hat{\beta}_{\mathbb{I},n}\right)
&= \int_{\mathbb{T}} f(x_{i1} |
\bm{\theta}_{i}) f(x_{i2} |
\bm{\theta}_{i})  \pi(\bm{\theta}_{i};\; \hat{\gamma}_{\mathbb{I},n}, \hat{\beta}_{\mathbb{I},n})
\text{d} \bm{\theta}_{i}\\
&= \binom{m_{i1}}{ x_{i1}} \binom{m_{i2}}{x_{i2}}
      \frac{B(x_{i1}+\hat{\gamma}_{\mathbb{I},n}, m_{i1} - x_{i1} +\hat{\beta}_{\mathbb{I},n})
B(x_{i2} +\hat{\gamma}_{\mathbb{I},n}, m_{i2} -
   x_{i2}+\hat{\beta}_{\mathbb{I},n})
                      }{\left[B(\hat{\gamma}_{\mathbb{I},n}, \hat{\beta}_{\mathbb{I},n})\right]^2},
\end{align*}
which gives
\[
  T_{\mathbb{I}}\left(\mathbf{D}_{n}\right) =
\frac{B(x_{i1}+\hat{\gamma}_{\mathbb{I},n}, m_{i1} - x_{i1} +\hat{\beta}_{\mathbb{I},n})
B(x_{i2} +\hat{\gamma}_{\mathbb{I},n}, m_{i2} -
   x_{i2}+\hat{\beta}_{\mathbb{I},n})}{[B(\hat{\gamma}_{\mathbb{I},n}, \hat{\beta}_{\mathbb{I},n})]^2
\tilde{\theta}_i^{x_{i1}+x_{i2} } (1 - \tilde{\theta}_i)^{m_{i1} +m_{i2} - x_{i1} -
   x_{i2}}},
  \]
  where $\hat{\gamma}_{\mathbb{I},n}$ and
  $\hat{\beta}_{\mathbb{I},n}$ are calculated in a similar fashion
  to Section \ref{sec:betabinomCI} except that data from both
  sequences should be used to estimate $\hat{\mu}_\text{EB}$ and
$\hat{\phi}_\text{EB}$, in the sense that
\begin{align*}
  \hat{\mu}_\text{EB} &= \frac{1}{2n-2} \sum_{k \neq i} \sum_{g=1}^{2} \frac{x_{kg}}{m_{kg}}, \mbox{ and}\\
  \hat{\phi}_\text{EB} &= \left[ \frac{\bar{m} V_{xy}}{ \hat{\mu}_\text{EB}( 1-\hat{\mu}_\text{EB})} -1 \right]\bigg{/} (\bar{m} - 1),
  \end{align*}
  where
\begin{align*}
  \bar{m} &= \frac{1}{2n-2}\sum_{k \neq i} \sum_{g=1}^{2} m_{kg}, \mbox{ and } \\
 V_{xy} &= \frac{1}{2n-2} \sum_{k \neq i} \sum_{g=1}^{2}
                       \left(\frac{x_{kg}}{m_{kg}} -  \hat{\mu}_\text{EB}\right)^2.
   \end{align*}

In our first simulation, we assessed the performance of the test
statistic in terms of the Type I error. Assuming a window size of  $n=20$, realized
data $(x_{i1}, x_{i2})$ ($i\in[n]$),  were simulated from independent binomial
distributions with $\theta_{i1}^{*}= \theta_{i2}^{*} =\theta_i^{*}$ ($i = 1,
\ldots, n$), with $\theta_i^{*}$ 
ranging between $0.1$ and $0.9$, and $m_{i1}, m_{i2}\in\mathbb{N}$ uniformly and independently
sampled from the range
$[15, 40]$. The first panel of Figure \ref{fig-betabinom-Test1}  shows
the calculated
test statistic values $T_{\mathbb{I}}\left(\mathbf{D}_{n}\right)$ for the $20$ genomic indices on the
logarithmic scale, over 100 independently
replicated datasets, with horizontal lines displaying values $\log(1/\alpha)$, for
significance levels 
$\alpha \in \{0.01, 0.02, 0.05\}$. No points were observed above the line
corresponding to $\alpha=0.01$, indicating that the Type I error of the
test statistic does not exceed the nominal level.
Next, we assessed the power of the test statistic
  at three levels of significance
($\alpha \in \{ 0.01, 0.02, 0.05\}$) and differing effect sizes. For each $i$ ($i\in[n]$),  $\theta_{i1}^{*}$ was set to be a value
between $0.05$ and $0.95$, and $\theta_{i2}^{*} = \theta_{i1}^{*}+\Delta$, where $0.1<\Delta<0.9$  (with $\theta_{i2}^{*} <1$). A sample of $20$
replicates were simulated under each possible set of values of
$(\theta_1^{*}, \theta_2^{*})$. The second panel of Figure \ref{fig-betabinom-Test1} shows that the power functions increased
rapidly to 1 as the difference $\Delta$ was increased.

\begin{figure}
  \begin{center}
   \includegraphics[width=0.45\textwidth]{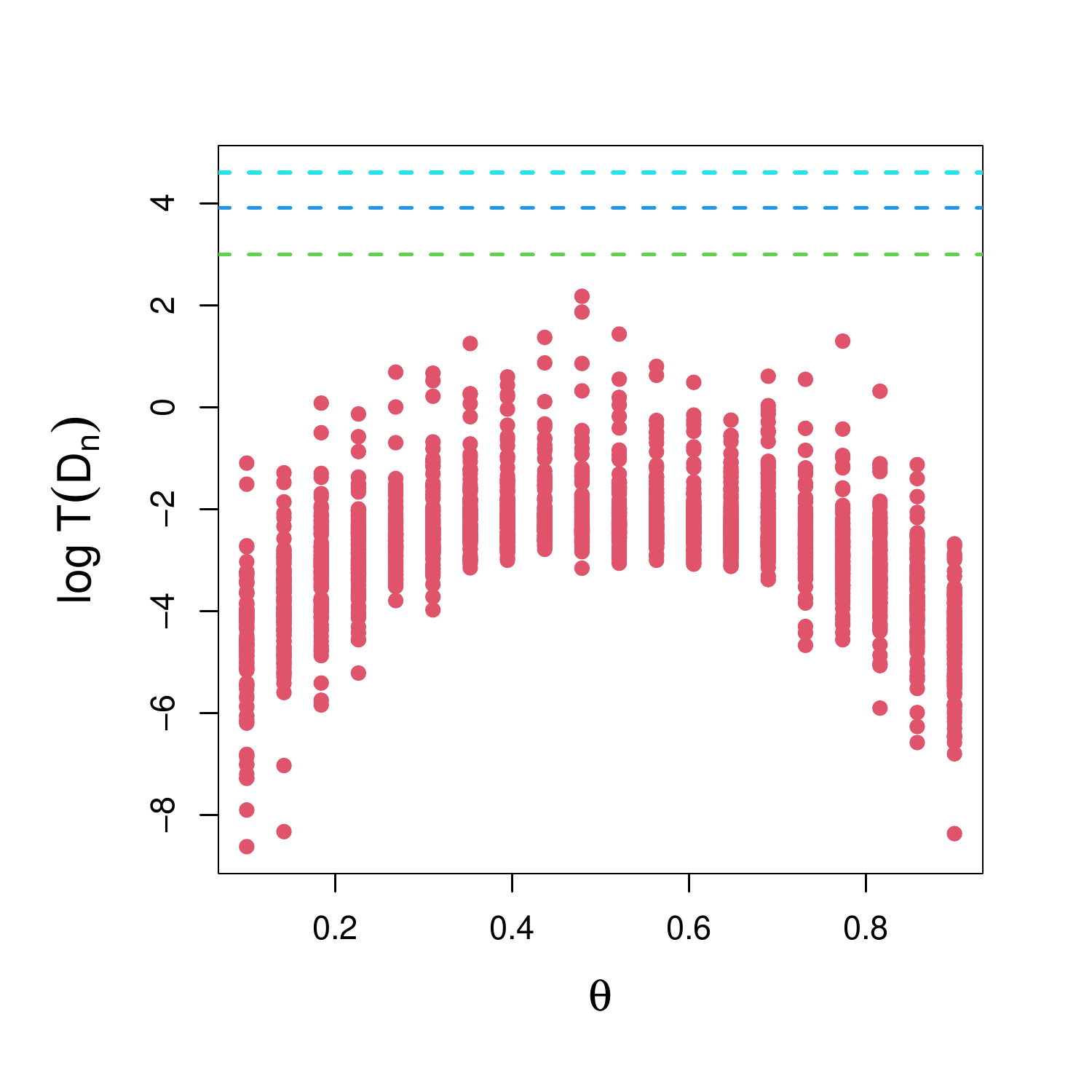}
   \includegraphics[width=0.45\textwidth]{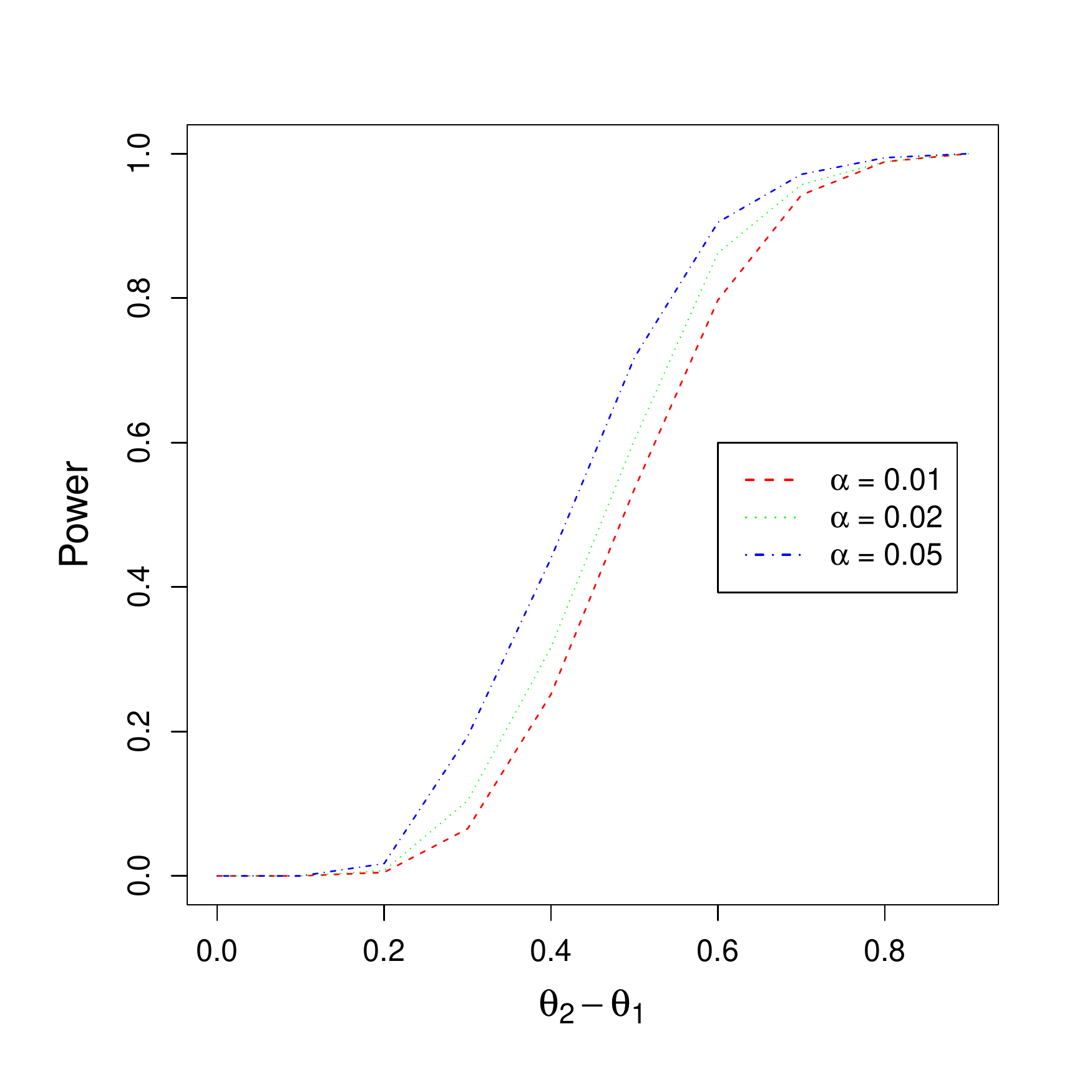}

    \caption{\label{fig-betabinom-Test1}  Panel (a): Test statistic for 100
      replications of the beta--binomial example under the null
      hypothesis of equality of proportions. The three horizontal lines correspond to cutoffs according to significance levels of $\alpha= 0.05$ (green), $\alpha= 0.02$ (blue), and $\alpha= 0.01$ (turquoise).
       Panel (b): Power function
      over different values of $\Delta =\theta_2^{*} -\theta_1^{*}$ at three
       levels of significance: $\alpha \in \{0.01, 0.02, 0.05\}$.}
    \end{center}
  \end{figure}
 
 In our next numerical experiment, we generated data sets of sizes $n \in \{10, 100, 1000 \}$, where realized
observations $x_{i1}$, and $x_{i2}$ are simulated from independent binomial
distributions with parameters $\theta_{i1}^{*}$ and $\theta_{i2}^{*}$, respectively ($i\in[n]$). 
For each $i$, $\theta_{i1^{*}}$ was generated from a beta distribution, in turn, with hyperparameters $\bm{\psi} = (\gamma, \beta) \in \{(2,2), (2,5), (5,2)\}$; and $\theta_{i2}^{*} = \theta_{i1}^{*} + \Delta$, where $\Delta  \in \{0, 0.2, 0.5, 0.9\}$. We generated 100 instances of data under each setting and assessed the power of the FSEB test statistic through the number of rejections at levels  $\alpha \in \{ 0.0005, 0.005, 0.05 \}$. The results are shown in Table \ref{tab:beta-binomial sim - test}. 

Similarly to the Poisson--gamma example, it can be seen that the tests reject true null hypotheses at below the nominal sizes $\alpha$, in each case.  For each combination of $n$ and $\bm{\psi}$, as $\Delta$ increases, the rejection rate increases, making the tests more powerful as expected, when detecting larger differences between $\theta_{i1}^{*}$ and $\theta_{i2}^{*}$, frequently reaching a power of 1 even when the difference was not maximal. There did not appear to be a clear increase in power with the sample size, within the settings considered. Overall, we may conclude, as previously, that the tests are behaving as expected, although both this example and the Poisson--gamma case show that the tests may be underpowered as they do not achieve the nominal size for any value of $\alpha$. 
 
\begin{table}
\caption{\label{tab:beta-binomial sim - test}Experimental results
for testing the hypothesis $\text{H}_{0}$: $\theta_{i1}^{*}=\theta_{i2}^{*}$
for Beta--binomial count series models. The Rejection proportion columns
report the average number of rejections, from 100 test replicates, at levels
of significance $\alpha\in\left\{ 0.05,0.005,0.0005\right\} $.}

\centering{}
\begin{tabular}{lcllll}
  \hline
   &   &   & \multicolumn{3}{c}{Rejection proportion at level $\alpha$} \\
$n$ & $\bm{\psi}$ &$\Delta$ & $0.0005$ & $0.005$ & $0.05$  \\ 
  \hline
  \hline 
  10 & $(2,2)$& 0 & 0.000 & 0.000 & 0.000 \\ 
  & & 0.2 & 0.000 & 0.004 & 0.039 \\ 
  & & 0.5 & 0.305 & 0.471 & 0.709 \\ 
  & & 0.9 & 0.980 & 1.000 & 1.000 \\ 
  &$(2,5)$ & 0 & 0.000 & 0.000 & 0.000 \\ 
  & & 0.2 & 0.000 & 0.001 & 0.025 \\ 
  & & 0.5 & 0.249 & 0.464 & 0.692 \\ 
  & & 0.9 & 0.995 & 1.000 & 1.000 \\ 
  &$(5,2)$ & 0 & 0.000 & 0.000 & 0.000 \\ 
  & & 0.2 & 0.000 & 0.006 & 0.052 \\ 
  & & 0.5 & 0.281 & 0.459 & 0.690 \\ 
  & & 0.9 & 0.993 & 0.993 & 1.000 \\ 
\hline 
  100 &$(2,2)$ & 0 & 0.000 & 0.000 & 0.000 \\ 
  & & 0.2 & 0.000 & 0.004 & 0.037 \\ 
  & & 0.5 & 0.272 & 0.459 & 0.700 \\ 
  & & 0.9 & 0.996 & 0.998 & 1.000 \\ 
  &$(2,5)$ & 0 & 0.000 & 0.000 & 0.000 \\ 
  & & 0.2 & 0.000 & 0.003 & 0.032 \\ 
  & & 0.5 & 0.267 & 0.459 & 0.693 \\ 
  & & 0.9 & 0.994 & 0.999 & 1.000 \\ 
  &$(5,2)$ & 0 & 0.000 & 0.000 & 0.000 \\ 
  & & 0.2 & 0.000 & 0.004 & 0.047 \\ 
  & & 0.5 & 0.269 & 0.459 & 0.697 \\ 
  & & 0.9 & 0.987 & 0.998 & 0.999 \\ 
  \hline 
  1000 &$(2,2)$ & 0 & 0.000 & 0.000 & 0.000 \\ 
  & & 0.2 & 0.000 & 0.003 & 0.031 \\ 
  & & 0.5 & 0.280 & 0.476 & 0.707 \\ 
  & & 0.9 & 0.982 & 0.992 & 0.998 \\ 
  &$(2,5)$ & 0 & 0.000 & 0.000 & 0.000 \\ 
  & & 0.2 & 0.000 & 0.003 & 0.030 \\ 
  & & 0.5 & 0.264 & 0.459 & 0.693 \\ 
  & & 0.9 & 0.989 & 0.996 & 1.000 \\ 
  &$(5,2)$ & 0 & 0.000 & 0.000 & 0.000 \\ 
  & & 0.2 & 0.000 & 0.005 & 0.047 \\ 
  & & 0.5 & 0.279 & 0.474 & 0.706 \\ 
  & & 0.9 & 0.986 & 0.995 & 0.999 \\ 
   \hline
\end{tabular}
\end{table}

As an additional assessment of how FSEB performs in comparison to other tests in a similar setting, we carried out a number of additional simulation studies, in which FSEB was compared with Fisher's exact test and a score test,  over various settings of $n$, $\bm{\psi}$ and $\Delta$, as well as for different ranges of $m_i$ ($i=1\in[n]$). Comparisons were made using the $p$-values as well as false discovery rate (FDR) corrected $p$-values arising from FDR control methods \citep{wang2020false}, and are presented in the online Supplementary Materials (Tables S1--S8 and Figures S1--S8). It is evident in almost all cases (and especially in case C, which most closely resembles the real life application scenario) that (i) the power levels are very similar across methods, especially as values of $n$, $m_i$ ($i\in[n]$) and effect sizes increase, and (ii) in every case, there are some settings in which Fisher's test and the score test are anti-conservative (even after FDR correction), with their Type I error greatly exceeding the nominal levels of significance, while this never occurs for FSEB, even without FDR correction.

\section{\label{sec:Results-applications}Real-data applications} 

\subsection{The \texttt{Norberg} data}

We now wish to apply the FSEB CI construction from Section \ref{subsec:Confidence-intervals-Poisson}
to produce CIs in a real data application. We shall investigate the
$\texttt{Norberg}$ data set from the $\textsf{REBayes}$ package
of \citet{Koenker:2017tq}, obtained from \citet{Haastrup2000Comparison-of-s}.
These data pertain to group life insurance claims from Norwegian workmen.
Here, we have $n=72$ observations $\mathbf{D}_{n}$, containing total
number of death claims $X_{i}$, along with covariates $\mathbf{w}_{n}$,
where $w_{i}$ is the number of years of exposure, normalized by a
factor of 344, for $i\in\left[n\right]$. Here each $i$ is an individual
occupation group.

To analyze the data, we use the Poisson--gamma model and estimate
the generative parameters $\bm{\vartheta}_{n}^{*}$ using estimates
of form \eqref{eq: poisson--gamma estimator}. Here, each $\theta_{i}^{*}$
can be interpreted as an unobserved multiplicative occupation specific
risk factor that influences the number of claims made within occupation
group $i$. To obtain \textcolor{black}{individually-valid} $95\%$ CIs for each of the $n$ estimates, we then apply
the method from Section \ref{subsec:Confidence-intervals-Poisson}.
We present both the estimated risk factors and their CIs in Figure
\ref{fig:Estimates-of-risk}.

\begin{figure}
\begin{centering}
\includegraphics[width=15cm]{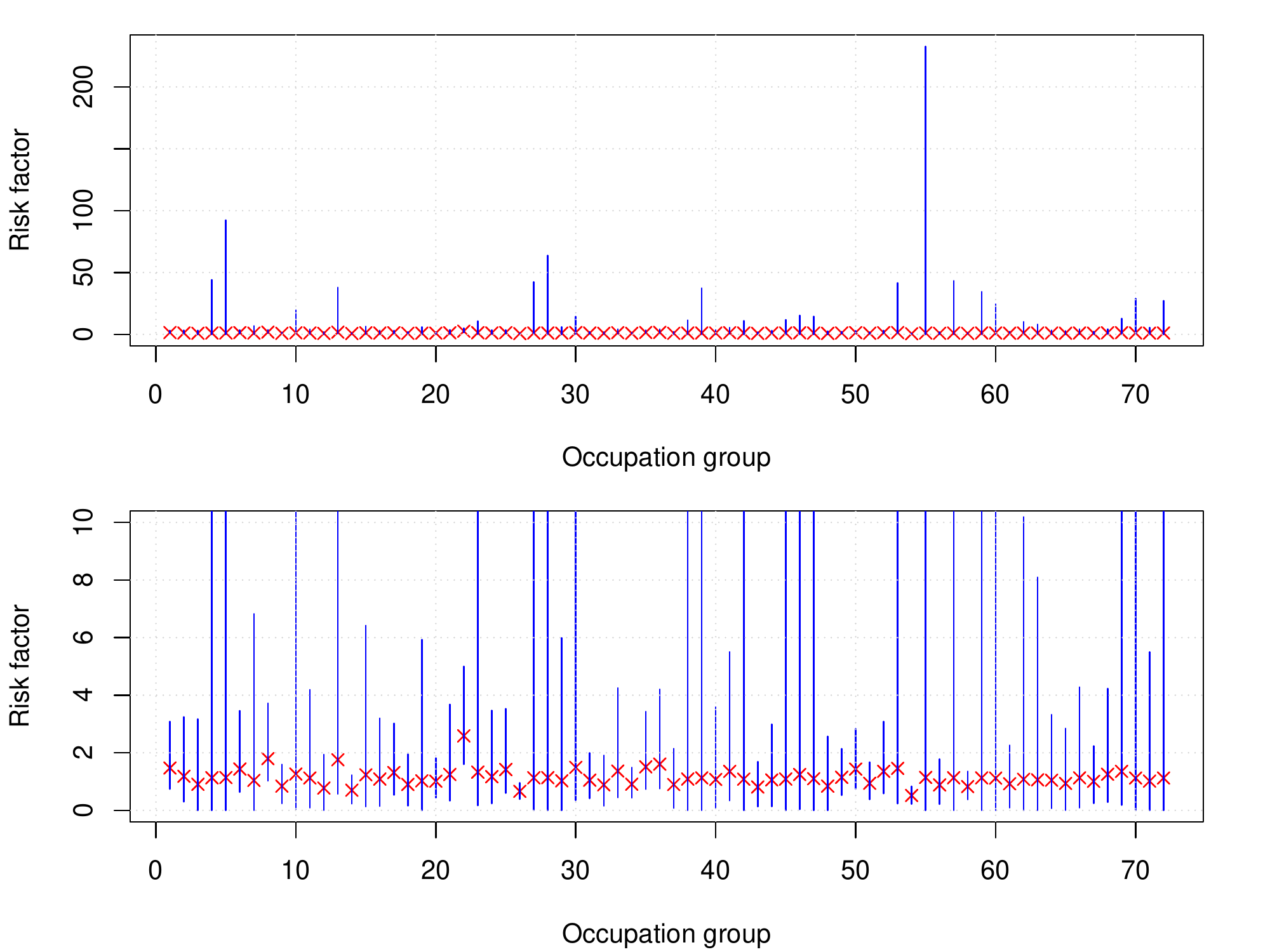}
\par\end{centering}
\caption{\label{fig:Estimates-of-risk}Estimates of risk factors $\bm{\vartheta}_{n}^{*}$
for the $\texttt{Norberg}$ data set along with associated $95\%$ CIs. The
estimated risk factor for each occupation group is depicted as a cross
and the associate \textcolor{black}{(individually-valid)} CI is depicted as a line. The top plot displays
the CIs at their entire lengths, whereas the bottom plot displays
only the risk factor range between 0 and 10.}

\end{figure}

From Figure \ref{fig:Estimates-of-risk}, we notice that most of the
estimates of $\bm{\vartheta}_{n}^{*}$ are between zero and two, with
the exception of occupation group $i=22$, which has an estimated
risk factor of $\theta_{22}^{*}=2.59$. Although the risk factors
are all quite small, the associated CIs can become very large, as can
be seen in the top plot. This is due to the conservative nature of
the CI constructions that we have already observed from Section \ref{subsec:Stein's-problem}.

We observe that wider CIs were associated with observations where $X_{i}=0$, with $w_{i}$ being small. In particular, the largest CI, occurring for $i=55$, has response $X_{55}=0$ and the smallest covariate value in the data set: $w_{55}=4.45$. The next largest CI occurs for $i=5$ and also corresponds to a response $X_{5}=0$ and the second smallest covariate value $w_{5}=11.30$.

However, upon observation of the bottom plot, we see that although
some of the CIs are too wide to be meaningful, there are still numerous
meaningful CIs that provide confidence regarding the lower limits
as well as upper limits of the underlying risk factors. In particular,
we observe that the CIs for occupation groups $i=26$ and $i=54$
are remarkably narrow and precise. \textcolor{black}{Of course, the preceding inferential observations are only valid when considering each of the $n$ CIs, individually, and under the assumption that we had chosen to draw inference regarding the corresponding parameter of the CI, before any data are observed.} 

\textcolor{black}{If we wish to draw inference regarding all $n$ elements of $\bm{\vartheta}_{n}^{*}$,
simultaneously, then we should instead construct a $100\left(1-\alpha\right)\%$
simultaneous confidence set $\bar{\mathcal{C}}^{\alpha}\left(\mathbf{D}_{n}\right)$,
with the property that
\[
\text{Pr}_{\bm{\vartheta}_{n}^{*}}\left[\bm{\vartheta}_{n}^{*}\in\bar{\mathcal{C}}^{\alpha}\left(\mathbf{D}_{n}\right)\right]\ge1-\alpha\text{.}
\]
Using Bonferroni's inequality, we can take $\bar{\mathcal{C}}^{\alpha}\left(\mathbf{D}_{n}\right)$
to be the Cartesian product of the individual $100\left(1-\alpha/n\right)\%$
(adjusted) CI for each parameter $\theta_{i}^{*}$:
\[
\bar{\mathcal{C}}^{\alpha}\left(\mathbf{D}_{n}\right)=\prod_{i=1}^{n}\mathcal{C}_{i}^{\alpha/n}\left(\mathbf{D}_{n}\right)\text{.}
\]
Using the $\alpha=0.05$, we obtain the $95\%$ simultaneous confidence
set that appears in Figure \ref{fig: Simultaneous Norberg}. We observe
that the simultaneous confidence set now permits us to draw useful
inference regarding multiple parameters, at the same time. For example,
inspecting the $n$ adjusted CIs, we observe that the occupations
corresponding to indices $i\in\left\{ 8,22,50\right\} $ all have
lower bounds above $0.5$. Thus, interpreting these indices specifically,
we can say that each of the three adjusted confidence intervals, which
yield the inference that the risk factors $\theta_{i}^{*}>0.5$ for
$i\in\left\{ 8,22,50\right\} $, contains the parameter $\theta_{i}^{*}$
with probability $0.95$, under repeated sampling.}

\begin{figure}
\begin{centering}
\includegraphics[width=15cm]{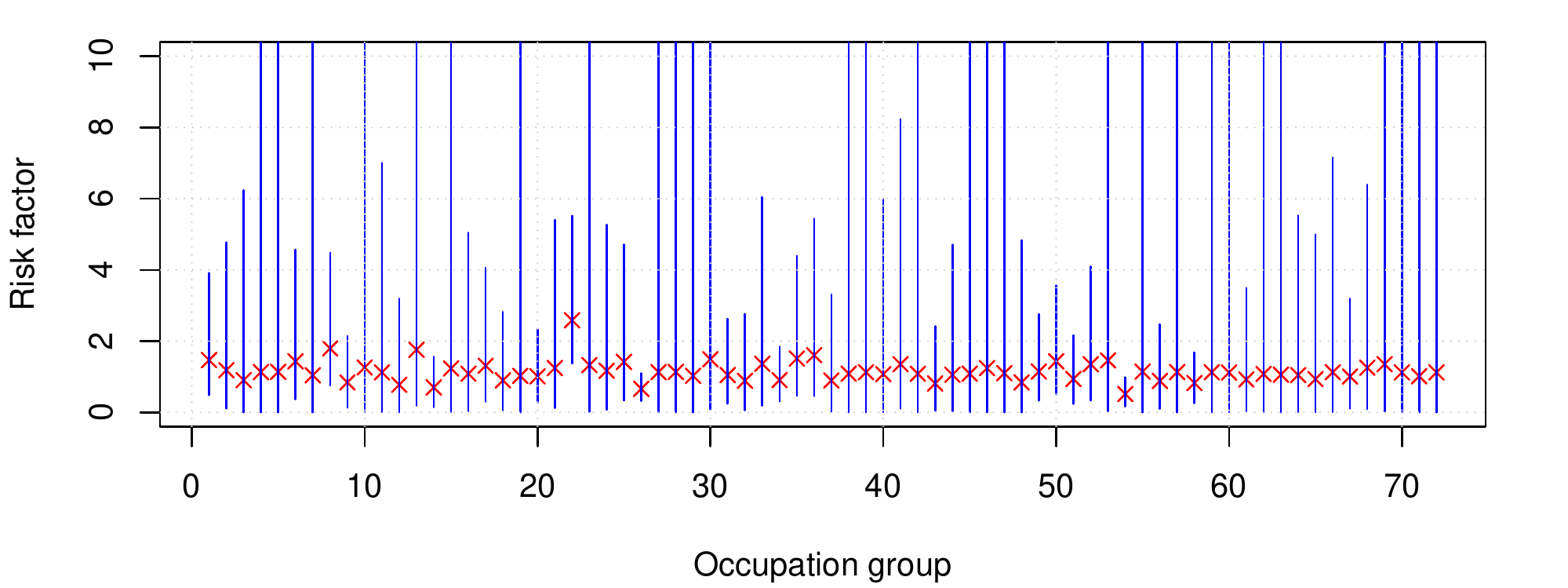}
\par\end{centering}
\caption{\label{fig: Simultaneous Norberg} \textcolor{black}{Estimates of risk factors $\bm{\vartheta}_{n}^{*}$
for the \texttt{Norberg} data set along with the associated simultaneous
$95\%$ confidence set. The estimated risk factors for each occupation
group is depicted as a cross and the simultaneous confidence set can
be constructed via the cartesian product of the adjusted CIs, depicted
as lines. The plot is focused on the risk factor range between 0 and
10.}}

\end{figure}

\textcolor{black}{Since our individual CI and adjusted CI constructions are $e$-CIs,
one can alternatively approach the problem of drawing simultaneously
valid inference via the false coverage rate (FCR) controlling techniques
of \cite{xu2022post}. Using again the parameters $\theta_{i}^{*}$ corresponding
to $i\in\left\{ 8,22,50\right\} $, as an example, we can use Theorem
2 of \cite{xu2022post} to make the statement that the three adjusted CIs $\mathcal{C}_{i}^{3\alpha/n}\left(\mathbf{D}_{n}\right)$,
for $i\in\left\{ 8,22,50\right\} $, can be interpreted at the FCR
controlled level $\alpha\in(0,1)$, in the sense that
\[
\text{E}_{\bm{\vartheta}_{\mathbb{I}\left(\mathbf{D}_{n}\right)}^{*}}\left[\frac{\sum_{i\in\mathbb{I}}\left\llbracket \theta_{i}^{*}\notin\mathcal{C}_{i}^{\left|\mathbb{I}\left(\mathbf{D}_{n}\right)\right|\alpha/n}\left(\mathbf{D}_{n}\right)\right\rrbracket }{\max\left\{ 1,\left|\mathbb{I}\left(\mathbf{D}_{n}\right)\right|\right\} }\right]\le\alpha\text{,}
\]
where $\mathbb{I}\left(\mathbf{D}_{n}\right)$ is a data-dependent
subset of parameter indices. In particular, we observe the realization
$\left\{ 8,22,50\right\} $ of $\mathbb{I}\left(\mathbf{D}_{n}\right)$,
corresponding to the data-dependent rule of selecting indices with
adjusted CIs $\mathcal{C}_{i}^{\alpha/n}\left(\mathbf{D}_{n}\right)$
with lower bounds greater than $0.5$. Here, $\left\llbracket \mathsf{A}\right\rrbracket =1$
if statement $\mathsf{A}$ is true and $0$, otherwise.}

\textcolor{black}{Clearly, controlling the FCR at level $\alpha$ yields narrower CIs
for each of our the three assessed parameters than does the more blunt
simultaneous confidence set approach. In particular, the $95\%$ simultaneous
adjusted CIs obtained via Bonferroni's inequality are $\left(0.775,4.485\right)$,
$\left(1.375,5.520\right)$, and $\left(0.505,3.565\right)$, and
the $0.05$ level FCR controlled adjusted CIs are $\left(0.810,4.300\right)$,
$\left(1.430,5.390\right)$, and $\left(0.555,3.390\right)$, for
the parameters $\theta_{i}^{*}$ corresponding to the respective parameters
$i\in\left\{ 8,22,50\right\} $. Overall, these are positive results as we
do not know of another general method for generating CIs in this EB setting, whether individually or jointly.}

\subsection{\label{sec:betabinom-applic} Differential
  methylation detection in bisulphite sequencing data }

DNA methylation is a chemical modification of DNA caused by the
addition of a methyl ($CH_3$-) group to a DNA nucleotide -- usually a C that is followed by a G -- called a CpG site, which is an important factor in controlling gene expression over the human
genome. Detecting differences in the methylation patterns between normal and
ageing cells can shed light on the complex biological processes underlying human
ageing, and hence has been an important scientific problem
over the last decade \citep{Smith13}. Methylation patterns can be detected using high-throughput
bisulphite sequencing experiments \citep{Krueger12}, in which data are
generated in the form of sequences of numbers of methylated cytosines, $x_{ig}$, among the total counts of cytosines, $m_{ig}$, for $n$  CpG sites on a  genome $(i\in[n])$, for $G$ groups of cell types $g\in[G]$. Often, 
there are $G=2$ groups, as
in our example that follows, 
for which the question of interest is to detect regions of
differential methylation in the DNA of normal and ageing cells. Based on the setup above, a set of bisulphite sequencing data from an experiment with $G$ groups might be considered as $G$ series of (possibly correlated) observations from non-identical binomial distributions. The degree of dependence between adjacent CpG sites typically depends on the genomic distance between these loci, but since these are often separated by hundreds of bases,  for the moment it is assumed that this correlation is negligible and is not incorporated into our model.

  \subsubsection{Application to Methylation data from Human chromosome 21}

We evaluated the test statistic
$T_{\mathbb{I}}\left(\mathbf{D}_{n}\right)$ over a paired segment of
methylation data from  normal and ageing cells, from 
$100,000$ CpG sites on human chromosome 21
\citep{Cruickshanks13}. After data cleaning and filtering (to remove
sites with too low or too high degrees of experimental coverage, that
can introduce errors), $58,361$ sites remained
for analysis.
Figure \ref{fig-betabinom-methyl} shows the predicted  demarcation of the data
into differentially and
non-differentially methylated sites over the entire
region, at three cutoff levels of significance, overlaid with a
moving average  using a window size of 10 sites. It was observed that
large values of the test statistic
were often found in grouped clusters, which would be biologically
meaningful, as loss of methylation in ageing cells is more likely
to be highly region-specific, rather than randomly scattered over the
genome. The overall rejection
rates for the FSEB procedure  corresponding to significance levels of $\alpha= 0.0005, 0.05, 0.02$
and $0.01$ were found to be $0.0012$, 
$0.0154$, $0.0092$, and $ 0.0064$, respectively.

\begin{figure}
  \begin{center}
    \includegraphics[width=\textwidth]{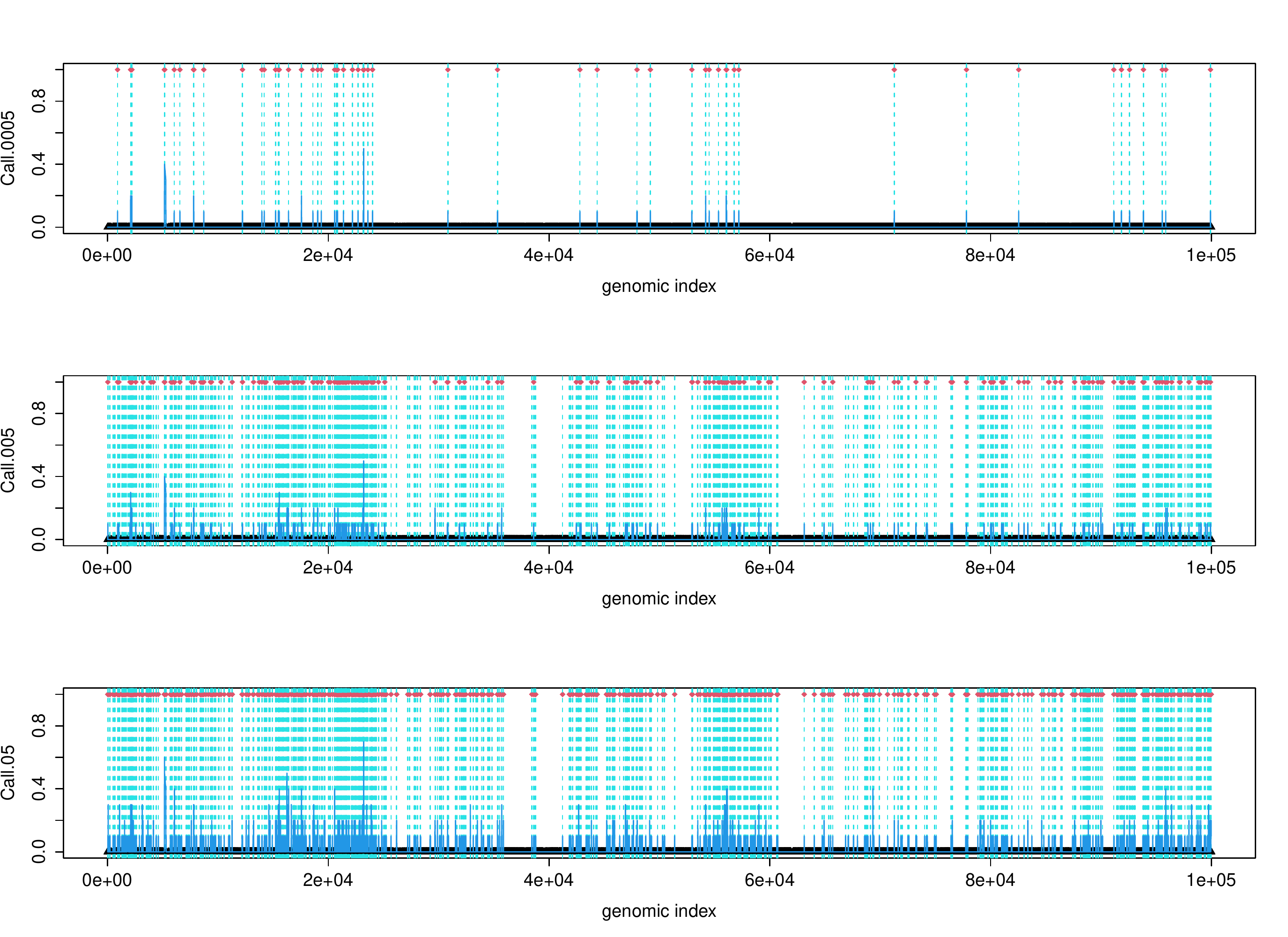} 
   
    \caption{FSEB test statistics over a segment of methylation data. The
      panels show the demarcation of loci into differentially
      methylated (coded as ``$1$'') and non-differentially methylated sites (coded as ``$0$'') 
with an overlay of a
moving average with a window size of 10 CpG sites, at significance
level cutoffs of $0.0005$, $0.005$, and $0.05$.
}\label{fig-betabinom-methyl}
    \end{center}
  \end{figure}

As a comparison to other methods for detecting
differential methylation, we also applied site-by-site Fisher tests
and score tests as implemented for bisulphite sequencing data in the {\sf R} Bioconductor package {\sf
  DMRcaller} \citep{DMRcallerNAR2018}. For purposes of comparison, we used
two significance level cutoffs of $0.05$ and $0.0005$ for our FSEB test
statistic, along with the same cutoffs subject to a  Benjamini--Hochberg
FDR correction for the other two testing methods. 
Figure \ref{fig-methyl-comp} shows the comparison between the
calculated site-specific $p$-values of the Fisher and score tests with
the calculated FSEB test statistic (all on the logarithmic scale) over
the entire genomic segment, which indicates a remarkable degree
of overlap in the regions of differential methylation. There are,
however, significant differences as well, in both the numbers of
differential methylation calls and their location. In particular, the
FSEB test statistic appeared to have stronger evidence for differential
methylation in two regions, one on the left side of the figure, 
and one towards the centre. The Fisher test, being the most conservative, almost missed this central
region (gave a very weak signal), while the score test gave a very
high proportion of differential methylation calls compared to both
other methods -- however, the results from the score test may not be as reliable as many cells contained small numbers of counts which may render the test assumptions invalid.
Table \ref{tab-methyl-comp} gives a summary of the overlap and
differences of the results from
the different methods at two levels of significance, indicating that with FDR corrections,  the Fisher test appears to be the most conservative, the score test
the least conservative, and the FSEB procedure in-between the two. 
We also calculated, for each pair of methods, the 
proportion of matching calls, defined as the ratio of the number of sites predicted by both methods as either differentially methylated, or non-differentially methylated, to the total number of sites. These proportions indicated a high degree of concordance, especially between FSEB and Fisher tests, with the score test showing the least degree of concordance at both levels of significance. As expected, the degree of concordance decreased with an increase in $\alpha$, but only slightly so, between the FDR-corrected Fisher test and FSEB.

\begin{figure}
  \begin{center}
    \includegraphics[width=\textwidth]{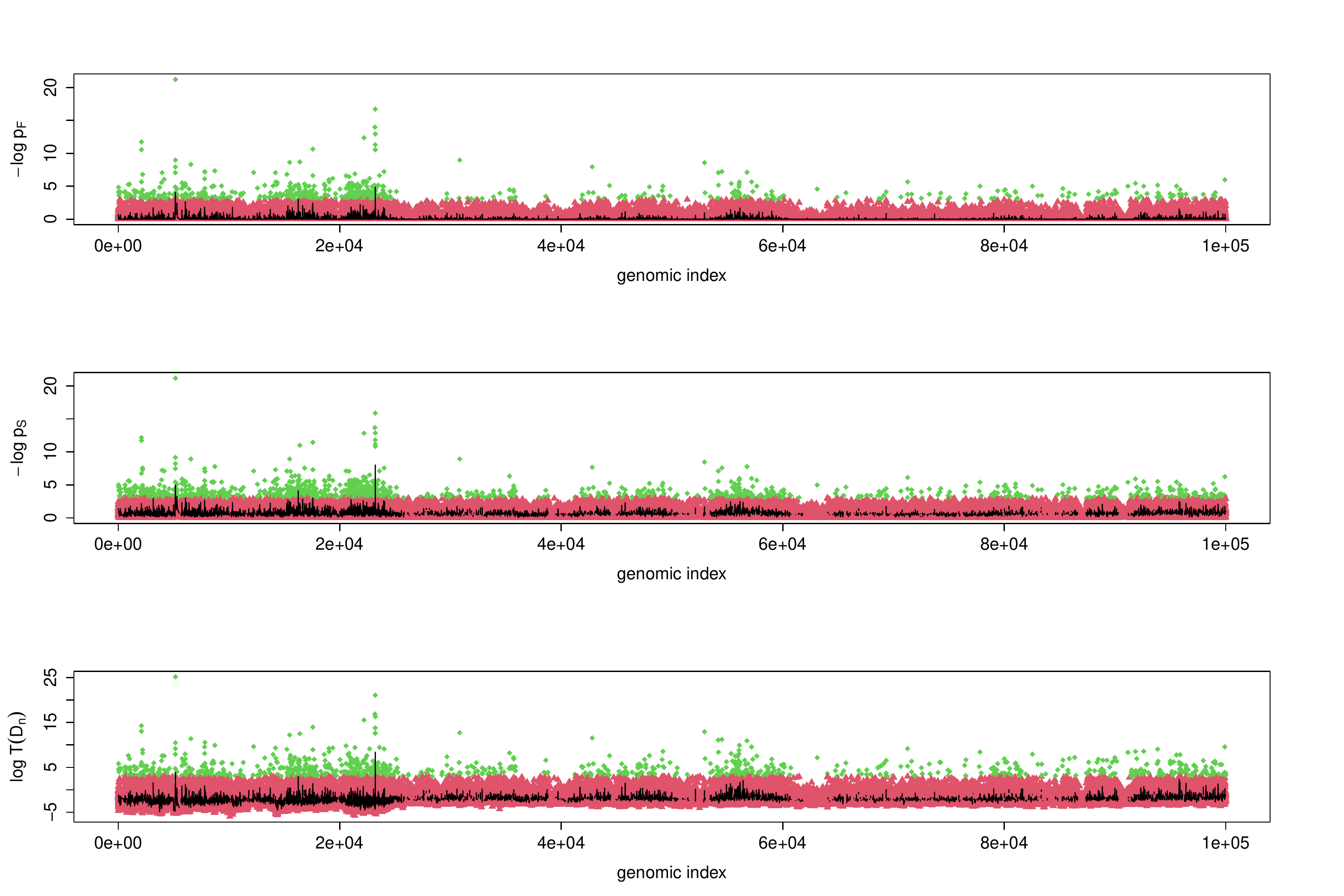} 
   
    \caption{Results of three testing procedures to detect sites of
      differential methylation over a segment of methylation data. The
      first two
      panels show the negative logarithms of the FDR-corrected $p$-values
      for the (i) Fisher test ($-\log p_F$) and (ii) score test ($-\log p_S$), while the third
      panel shows the logarithm of the FSEB test statistic ($\log T(D_n)$). The black
     curve in each plot corresponds to a
      moving average with a window size of 10.  The points are coloured
      by differential methylation state call: green if differentially
      methylated, and red if not, at test size
      $0.05$. 
}\label{fig-methyl-comp}
    \end{center}
  \end{figure}

\begin{table}
\caption{\label{tab-methyl-comp}   Comparison of differential
  methylation calling results between different methods: (i) FSEB (ii)
  Fisher tests with FDR-adjusted $p$-values (FF) (iii) Fisher tests,
  unadjusted (F)
  (iv) score tests with FDR-adjusted $p$-values (SF) and (v) score
  tests,  unadjusted (S).
  The upper table gives the
  proportions of sites called to be differentially expressed under the tests of
 sizes $\alpha \in \{0.0005, 0.05\}$. The lower table gives the proportion of
overlaps between differential methylation calls from each pair of
methods at a fixed level  $\alpha \in \{0.0005, 0.05\}$.}
\begin{centering}
  \begin{tabular}{lcc}
    \hline 
  &  \multicolumn{2}{c}{Proportion of rejections at level}\\ 
    & $\alpha= 0.0005$ & $\alpha= 0.05$ \\
\hline 
  \hline
  FSEB &  0.0012     &   0.0154   \\
FF &  0.0003  & 0.0097\\
F&   0.0098   & 0.1102\\
  SF &  0.1333    & 0.1528 \\
S &  0.1457    & 0.2926\\   
\hline                                                        
\end{tabular}
\par\end{centering}

\begin{centering}
  \begin{tabular}{lllll lllll}
    \hline 
 &\multicolumn{9}{c}{Proportion of overlap in matching calls at level }\\
&  \multicolumn{4}{c}{$\alpha=0.0005$} & & \multicolumn{4}{c}{$\alpha=0.05$} \\
\hline 
  \hline
  Method     &    FF  &    F    &     SF    &    S &    Method&                FF  &    F    &     SF    &    S\\
\hline
FSEB &    0.999 & 0.991 & 0.866 & 0.856 &        FSEB &           0.992  & 0.905  & 0.860  & 0.723\\
FF     &             & 0.991 & 0.867 & 0.855 &         FF     &             &        0.900  & 0.857  & 0.717\\
F     &               &           & 0.858 & 0.864 &          SF     &             &   &      0.777  & 0.818\\
SF       &              &         &        &       0.988 &           S       &               &   &   &             0.860\\
  \hline
\end{tabular}
\par\end{centering}
\end{table}

\section{\label{sec:Conclusion}Conclusion}

EB is a powerful and popular paradigm for conducting parametric inference
in situations where the DGP can be assumed to possess a hierarchical
structure. Over the years, general frameworks for point estimation
have been developed for EB, such as via the shrinkage estimators of
\cite{Serdobolskii2008Multiparametric} or the various method of moments and likelihood-based
methods described in \citet[Sec. 3]{Maritz:1989vp}. Contrastingly, the construction
of interval estimators and hypothesis tests for EB parameters rely
primarily on bespoke derivations and analysis of the specific models
under investigation.

In this paper, we have adapted the general universal inference framework
for finite sample valid interval estimation and hypothesis testing
of \cite{Wasserman:2020aa} to construct a general framework within the EB
setting, which we refer to as the FSEB technique. In Section 2, we proved
that these FSEB techniques generate valid confidence sets and hypothesis
tests of the correct size. In Section 3, we demonstrated via numerical
simulations, that the FSEB methods can be used in well-studied synthetic
scenarios. There, we highlight that the methods can generate meaningful
inference for realistic DGPs. This point is further elaborated in
Section 4, where we also showed that our FSEB approach can be usefully
applied to draw inference from real world data, in the contexts of
insurance risk and the bioinformatics study of DNA methylation.

We note that although our framework is general, due to it being Markov
inequality-based, it shares the same general criticism that may be
laid upon other universal inference methods, which is that the confidence
sets and hypothesis tests can often be conservative, in the sense
that the nominal confidence level or size is not achieved. 
\textcolor{black}{The lack of power due to the looseness of Markov's inequality was first mentioned and discussed in \cite{Wasserman:2020aa}, where it is also pointed out that, in the universal inference setting, the logarithm of the analogous ratio statistics to \eqref{eq: Ratio} have tail probabilities that scale, in $\alpha$, like those of $\chi^{2}$ statistics. The conservativeness of universal inference constructions is further discussed in the works of \cite{dunn2021gaussian}, \cite{tsenote}, and \cite{strieder2022choice}, where the topic is thoroughly explored via simulations and theoretical results regarding some classes of sufficiently regular problems.
We observe
this phenomenon in the comparisons in Sections 3.1 (and further expanded in the Supplementary
Materials). 
We also explored subsampling-based tests within the FSEB framework, along the lines proposed by \cite{dunn2021gaussian}, which led to very minor increases in power in some cases with small sample sizes without affecting the Type I error. With such an outcome not entirely discernible from sampling error, and with the substantial increase to computational cost, it does not seem worthwhile to employ the subsampling-based approach here. A possible reason for the lack improvement in power observed, despite subsampling, can be attributed to the fact that the sets $\mathbb{I}$, and their complements, are not exchangeable; since the indices fundamentally define the hypotheses and parameters of interest.
}

However, we note that since the methodology falls within
the $e$-value framework, it also inherits desirable properties, such
as the ability to combine test statistics by averaging \citep{Vovk:2021vi},
and the ability to more-powerfully conduct false discovery rate control
when tests are arbitrarily dependent \citep{wang2020false}.

Overall, we believe that FSEB techniques can be usefully incorporated
into any EB-based inference setting, especially when no other interval
estimators or tests are already available, and are a useful addition to the statistical tool set. Although a method that
is based on the careful analysis of the particular setting is always
preferable in terms of exploiting the problem specific properties
in order to generate powerful tests and tight intervals, FSEB methods
can always be used in cases where such careful analyses may be mathematically
difficult or overly time consuming.

\clearpage
\bibliographystyle{plainnat}
\bibliography{20210625_MASTERBIB}

\begin{thebibliography}{37}
\providecommand{\natexlab}[1]{#1}
\providecommand{\url}[1]{\texttt{#1}}
\expandafter\ifx\csname urlstyle\endcsname\relax
  \providecommand{\doi}[1]{doi: #1}\else
  \providecommand{\doi}{doi: \begingroup \urlstyle{rm}\Url}\fi

\bibitem[Ahmed and Reid(2001)]{AhmedReid:2001ws}
S~E Ahmed and N~Reid, editors.
\newblock \emph{{Empirical Bayes and Likelihood Inference}}.
\newblock Springer, New York, 2001.

\bibitem[Bickel(2020)]{Bickel:2020wp}
D~R Bickel.
\newblock \emph{{Genomics Data Analysis: False Discovery Rates and Empirical
  Bayes Methods}}.
\newblock CRC Press, Boca Raton, 2020.

\bibitem[Casella and Hwang(1983)]{Casella:1983ui}
G~Casella and J~T Hwang.
\newblock {Empirical Bayes confidence sets for the mean of a multivariate
  normal distribution}.
\newblock \emph{Journal of the American Statistical Association}, 78:\penalty0
  688--698, 1983.

\bibitem[Catoni et~al.(2018)Catoni, Tsang, Greco, and Zabet]{DMRcallerNAR2018}
M~Catoni, J~M Tsang, A~P Greco, and N~R Zabet.
\newblock {{D}{M}{R}caller: a versatile {R}/{B}ioconductor package for
  detection and visualization of differentially methylated regions in {C}p{G}
  and non-{C}p{G} contexts}.
\newblock \emph{Nucleic Acids Research}, 46:\penalty0 e114, 2018.

\bibitem[Cruickshanks et~al.(2013)Cruickshanks, McBryan, Nelson, Vanderkraats,
  Shah, van Tuyn, Rai, Brock, Donahue, Dunican, Drotar, Meehan, Edwards,
  Berger, and Adams]{Cruickshanks13}
H~A Cruickshanks, T~McBryan, D~M Nelson, N~D Vanderkraats, P~P Shah, J~van
  Tuyn, T~S Rai, C~Brock, G~Donahue, D~S Dunican, M~E Drotar, R~R Meehan, J~R
  Edwards, S~L Berger, and P~D Adams.
\newblock {S}enescent cells harbour features of the cancer epigenome.
\newblock \emph{Nature Cell Biology}, 15:\penalty0 1495--1506, 2013.

\bibitem[Datta et~al.(2002)Datta, Ghosh, Smith, and Lahiri]{Datta:2002wd}
G~S Datta, M~Ghosh, D~D Smith, and P~Lahiri.
\newblock {On an asymptotic theory of conditional and unconditional coverage
  probabilities of empirical Bayes confidence intervals}.
\newblock \emph{Scandinavian Journal of Statistics}, 29:\penalty0 139--152,
  2002.

\bibitem[Dunn et~al.(2021)Dunn, Ramdas, Balakrishnan, and
  Wasserman]{dunn2021gaussian}
R~Dunn, A~Ramdas, S~Balakrishnan, and L~Wasserman.
\newblock Gaussian universal likelihood ratio testing.
\newblock \emph{arXiv:2104.14676}, 2021.

\bibitem[Efron(2010)]{Efron2010}
B~Efron.
\newblock \emph{{Large-scale Inference}}.
\newblock Cambridge University Press, Cambridge, 2010.

\bibitem[Grunwald et~al.(2020)Grunwald, {de Heide}, and
  Koolen]{Grunwald:2020vf}
P~Grunwald, R~{de Heide}, and W~M Koolen.
\newblock {Safe testing}.
\newblock In \emph{{Information Theory and Applications Workshop (ITA)}}, 2020.

\bibitem[Haastrup(2000)]{Haastrup2000Comparison-of-s}
S~Haastrup.
\newblock {Comparison of some Bayesian analyses of heterogeneity in group life
  insurance}.
\newblock \emph{Scandinavian Actuarial Journal}, 2000:\penalty0 2--16, 2000.

\bibitem[Hardcastle and Kelly(2013)]{Hardcastle13}
T~J Hardcastle and K~A Kelly.
\newblock {{E}mpirical {B}ayesian analysis of paired high-throughput sequencing
  data with a beta-binomial distribution}.
\newblock \emph{BMC Bioinformatics}, 14:\penalty0 135, 2013.

\bibitem[Hwang and Zhao(2013)]{Hwang:2013vu}
J~T~G Hwang and Z~Zhao.
\newblock {Empirical Bayes confidence intervals for selected parameters in
  high-dimensional data}.
\newblock \emph{Journal of the American Statistical Association}, 108:\penalty0
  607--618, 2013.

\bibitem[Hwang et~al.(2009)Hwang, Qiu, and Zhao]{Hwang:2009vr}
J~T~G Hwang, J~Qiu, and Z~Zhao.
\newblock {Empirical Bayes confidence intervals shrinking both means and
  variances}.
\newblock \emph{Journal of the Royal Statistical Society B}, 71:\penalty0
  265--285, 2009.

\bibitem[Johnstone and Silverman(2005)]{Johnstone:2005uq}
I~M Johnstone and B~W Silverman.
\newblock {EbayesThresh: R programs for empirical Bayes thresholding}.
\newblock \emph{Journal of Statistical Software}, 12:\penalty0 1--38, 2005.

\bibitem[Kaufmann and Koolen(2018)]{Kaufmann2018Mixture-marting}
E~Kaufmann and W~M Koolen.
\newblock {Mixture martingales revisited with applications to sequential tests
  and confidence intervals}.
\newblock \emph{arXiv:1811.11419v1}, 2018.

\bibitem[Koenker and Gu(2017)]{Koenker:2017tq}
R~Koenker and J~Gu.
\newblock {REBayes: Empirical Bayes mixture methods in R}.
\newblock \emph{Journal of Statistical Software}, 82:\penalty0 1--26, 2017.

\bibitem[Krueger et~al.(2012)Krueger, Kreck, Franke, and Andrews]{Krueger12}
F~Krueger, B~Kreck, A~Franke, and S~R Andrews.
\newblock {DNA} methylome analysis using short bisulfite sequencing data.
\newblock \emph{Nature Methods}, 9:\penalty0 145--151, 2012.

\bibitem[Laird and Louis(1987)]{Laird:1987wo}
N~M Laird and T~A Louis.
\newblock {Empirical Bayes confidence intervals based on bootstrap samples}.
\newblock \emph{Journal of the American Statistical Association}, 82:\penalty0
  739--750, 1987.

\bibitem[Leng et~al.(2013)Leng, Dawson, Thomson, Ruotti, Rissman, Smits, Haag,
  Gould, Stewart, and Kendziorski]{Leng:2013ug}
N~Leng, J~A Dawson, J~A Thomson, V~Ruotti, A~I Rissman, B~M~G Smits, J~D Haag,
  M~N Gould, R~M Stewart, and C~Kendziorski.
\newblock {EBSeq: an empirical Bayes hierarchical model for inference in
  RNA-seq experiments}.
\newblock \emph{Bioinformatics}, 29:\penalty0 1035--1043, 2013.

\bibitem[Maritz and Lwin(1989)]{Maritz:1989vp}
J~S Maritz and T~Lwin.
\newblock \emph{{Empirical Bayes Methods}}.
\newblock CRC Press, Boca Raton, 1989.

\bibitem[Morris(1983{\natexlab{a}})]{Morris:1983uv}
C~N Morris.
\newblock {Parametric empirical Bayes inference: theory and applications}.
\newblock \emph{Journal of the American Statistical Association}, 78:\penalty0
  47--55, 1983{\natexlab{a}}.

\bibitem[Morris(1983{\natexlab{b}})]{morris1983parametric}
C~N Morris.
\newblock Parametric empirical bayes confidence intervals.
\newblock In \emph{Scientific inference, data analysis, and robustness}.
  Elsevier, 1983{\natexlab{b}}.

\bibitem[Narasimhan and Efron(2020)]{Narasimhan:2020ts}
B~Narasimhan and B~Efron.
\newblock {deconvolveR: a G-modeling program for deconvolution and empirical
  Bayes estimation}.
\newblock \emph{Journal of Statistical Software}, 94:\penalty0 1--20, 2020.

\bibitem[Norberg(1989)]{Norberg1989Experience-rati}
R~Norberg.
\newblock {Experience rating in group life insurance}.
\newblock \emph{Scandinavian Actuarial Journal}, 1989:\penalty0 194--224, 1989.

\bibitem[Serdobolskii(2008)]{Serdobolskii2008Multiparametric}
V~I Serdobolskii.
\newblock \emph{{Multiparametric Statistics}}.
\newblock Elsevier, Amsterdam, 2008.

\bibitem[Shafer(2021)]{Shafer:2021vh}
G~Shafer.
\newblock {Testing by betting: a strategy for statistical and scientific
  communication}.
\newblock \emph{Journal of the Royal Statistical Society B}, 184:\penalty0
  407--431, 2021.

\bibitem[Smith and Meissner(2013)]{Smith13}
Z~D Smith and A~Meissner.
\newblock {DNA} methylation: roles in mammalian development.
\newblock \emph{Nature Reviews Genetics}, 14:\penalty0 204--220, 2013.

\bibitem[Stein(1956)]{Stein1956Inadmissibility}
C~Stein.
\newblock {Inadmissibility of the usual estimator for the mean of a
  multivariate normal distribution}.
\newblock In \emph{{Berkeley Symposium on Mathematical Statistics and
  Probability}}, 1956.

\bibitem[Strieder and Drton(2022)]{strieder2022choice}
D~Strieder and M~Drton.
\newblock On the choice of the splitting ratio for the split likelihood ratio
  test.
\newblock \emph{arXiv:2203.06748}, 2022.

\bibitem[Tai and Speed(2006)]{Tai:2006wm}
Y~C Tai and T~P Speed.
\newblock {A multivariate empirical Bayes statistic for replicated microarray
  time course data}.
\newblock \emph{Annals of Statistics}, 34:\penalty0 2387--2412, 2006.

\bibitem[Tse and Davison(2022)]{tsenote}
T~Tse and A~C Davison.
\newblock A note on universal inference.
\newblock \emph{Stat}, to appear, 2022.

\bibitem[Vovk(2007)]{Vovk:2007aa}
V~Vovk.
\newblock {Strong confidence intervals for autoregression}.
\newblock \emph{arXiv:0707.0660v1}, 2007.

\bibitem[Vovk and Wang(2021)]{Vovk:2021vi}
V~Vovk and R~Wang.
\newblock {E-values: calibration, combination, and applications}.
\newblock \emph{Annals of Statistics}, 49:\penalty0 1736--1754, 2021.

\bibitem[Wang and Ramdas(2022)]{wang2020false}
R~Wang and A~Ramdas.
\newblock False discovery rate control with e-values.
\newblock \emph{Journal of the Royal Statistical Society B}, 84:\penalty0
  822--852, 2022.

\bibitem[Wasserman et~al.(2020)Wasserman, Ramdas, and
  Balakrishnan]{Wasserman:2020aa}
L~Wasserman, A~Ramdas, and S~Balakrishnan.
\newblock {Universal inference}.
\newblock \emph{Proceedings of the National Academy of Sciences}, 117:\penalty0
  16880--16890, 2020.

\bibitem[Xu et~al.(2022)Xu, Wang, and Ramdas]{xu2022post}
Z~Xu, R~Wang, and A~Ramdas.
\newblock Post-selection inference for e-value based confidence intervals.
\newblock \emph{arXiv:2203.12572}, 2022.

\bibitem[Yoshimori and Lahiri(2014)]{Yoshimori:2014vx}
M~Yoshimori and P~Lahiri.
\newblock {A second-order efficient empirical Bayes confidence interval}.
\newblock \emph{Annals of Statistics}, 42:\penalty0 1233--1261, 2014.

\end{thebibliography}

\end{document}